\documentclass[twocolumn,preprintnumbers,amsmath,amssymb]{revtex4-2}
\usepackage[english]{babel}
\usepackage[utf8]{inputenc}
\usepackage{enumitem}
\usepackage{epsfig}
\usepackage{graphicx}
\usepackage{dcolumn}
\usepackage{bm}
\usepackage{amsmath,amsfonts}
\usepackage{amssymb}
\usepackage{amsthm}
\usepackage{physics}
\usepackage{comment}
\usepackage{color} 
\usepackage{xcolor}
\usepackage[final,colorlinks,linkcolor=blue,anchorcolor=blue,urlcolor=blue,citecolor=blue]{hyperref}
\usepackage{caption}
\usepackage{subcaption}
\usepackage{siunitx,booktabs}
\usepackage{epstopdf}
\usepackage{framed}
\usepackage{dsfont}
\usepackage{pifont}
\usepackage{float}
\newtheorem{theorem}{Theorem}
\newtheorem{proposition}{Proposition}
\newtheorem{definition}{Definition}
\newtheorem{lemma}{Lemma}

\DeclareMathOperator*{\argmax}{argmax}
\DeclareMathOperator*{\argmin}{argmin}
\begin{document}
	\title{Finite-key security of passive quantum key distribution}
	\author{Víctor Zapatero$^{1,2,3}$}
	\email{vzapatero@vqcc.uvigo.es}
	\author{Marcos Curty$^{1,2,3}$}
\affiliation{$^1$Vigo Quantum Communication Center, University of Vigo, Vigo E-36310, Spain}
\affiliation{$^2$Escuela de Ingeniería de Telecomunicación, Department of Signal Theory and Communications, University of Vigo, Vigo E-36310, Spain}
\affiliation{$^3$AtlanTTic Research Center, University of Vigo, Vigo E-36310, Spain}
\begin{abstract}
The passive approach to quantum key distribution (QKD) consists of eliminating all optical modulators and random number generators from QKD systems, in so reaching an enhanced simplicity, immunity to modulator side channels, and potentially higher repetition rates. In this work, we provide finite-key security bounds for a fully passive decoy-state BB84 protocol, considering a passive QKD source recently presented. With our analysis, the attainable secret key rate is comparable to that of the perfect parameter estimation limit, in fact differing from the key rate of the active approach by less than one order of magnitude. This demonstrates the practicality of fully passive QKD solutions.

\end{abstract}
\maketitle
\section{Introduction}\label{Introduction}
Quantum key distribution (QKD) allows to establish information-theoretically secure keys between remote locations through an insecure channel~\cite{PortmannRenner}. This security standard, which is a must to achieve long-term security guarantees, makes QKD a unique solution for private communications of the highest requisites. Since its conception in 1984~\cite{BB84}, QKD has experienced a remarkable progress. Nowadays, specialized companies supply integral QKD services~\cite{IdQuantique,Toshiba,qubridge,ThinkQuantum}, metropolitan QKD networks are being installed around the globe~\cite{SECOQC,FangXing,Sasaki,Stucki}, and a space-ground integrated QKD backbone with an extension of thousands of kilometres has been deployed~\cite{Chen}. 

Notably though, QKD security proofs rely on mathematical models that describe the behaviour of the QKD equipment, in so opening the door for mis-characterization loopholes. For this reason, strictly quantifying the level of security of QKD implementations is a thorny issue~\cite{Feihu}, and a critical vulnerability of real-life QKD systems is active modulation. Indeed, active modulators can be a source of information leakage in different ways, say, negligently encoding private information in undesired degrees of freedom, introducing correlations between adjacent pulses~\cite{Margarida,correlations,correlations_2}, or serving as a target for Trojan horse attacks (THAs)~\cite{Vakhitov,Gisin,Jain1,Jain2,Sajeed}. In the latter case, an eavesdropper (Eve) injects bright light pulses into a QKD system and measures the back-reflected light, possibly extracting information about the setting choices. Particularly, since the back-reflected light has passed through the modulators, it may be encoded with the same information as the signals prepared by the sender (Alice), or reveal the measurement basis selected by the receiver (Bob). In this regard, although one can model the information leakage and account for it in the estimation of the secret key length, this approach may severely affect the achievable performance of QKD according to the existing analyses~\cite{Lucamarini,Tamaki,Weilong,Navs}, unless sufficiently strong isolation is introduced. From this perspective, a better solution is provided by passive QKD, which eliminates all active modulation from QKD devices and replaces it by post-selection. Remarkably, not only this solution confers immunity to modulator side channels, but it may also simplify the hardware layout and boost the clock rate of a QKD system, at the price of reducing the secret key rate per pulse by roughly an order of magnitude.

As an example, one can use coherent light to passively generate decoy states in different ways~\cite{WCP_1,WCP_2,WCP_3,WCP_4} (see also the experimental reports in~\cite{experiment_1,experiment_2,experiment_3,experiment_4,experiment_5}), or to passively prepare random photon polarizations in a plane for a passive realization of the BB84 protocol~\cite{passive_BB84}. Notably as well, passive BB84 encoding and passive decoy-state preparation can be combined in a single linear optical setup, as recently shown in~\cite{Mike,Zapatero}. Indeed, these two papers were soon followed by pioneering experiments that prove the feasibility of fully passive QKD~\cite{Lu,Hu}, and in fact fully passive twin-field QKD has also been envisioned~\cite{Mike_2}.

Considering the passive source devised in~\cite{Zapatero}, in this work we present finite-key security bounds for a fully passive decoy-state BB84 protocol, carefully merging the encoding strategy of~\cite{Mike} with a sophistication of the parameter estimation method in~\cite{Zapatero}. Remarkably, the secret key rate per pulse attainable with our analysis and our passive scheme is closer than ever reported to the corresponding key rate with an active setup.

The structure of the paper goes as follows. In Sec.~\ref{Assumptions}, we list the assumptions that underlie our security analysis. In Sec.~\ref{Characterization}, we provide a characterization of the passive QKD source we consider. In Sec.~\ref{post-selection}, we describe the encoding scheme we adopt, together with an insightful entanglement-based picture that arises from it. In Sec.~\ref{Protocol_description}, we give a description of the passive decoy-state BB84 protocol that we contemplate. Sec.~\ref{decoy} contains all the details of the decoy-state method, and in Sec.~\ref{secret_key_parameters} we derive the necessary estimates for the secret key parameters. To finish with, Sec.~\ref{performance} illustrates the performance of the proposed analysis, and Sec.~\ref{Outlook} ends up with a series of concluding remarks. For the reproducibility of the results, various appendices are included at the end of the paper as well.
\section{Assumptions}\label{Assumptions}
To begin with, let us enumerate various assumptions on which our analysis relies.\\

\textbf{Assumptions on Alice's and Bob's devices.} For the characterization of the fully passive source, we assume (i) perfect phase-randomization of the input laser/s, (ii) perfect-visibility interference at the beam splitters and polarizing beam splitters, and (iii) noiseless polarization and intensity measurements in Alice's side. As for Bob's measurement unit, we adopt the standard basis-independent detection efficiency assumption. Note, however, that this latter assumption could be removed by considering a measurement-device-independent configuration~\cite{MDI}.\\

\textbf{Adversary model.} We assume that Eve holds a quantum probe, $E$, and a classical register, $R$, and consider the following round-by-round model. In the first round, she initializes her probe $E$ to an arbitrary state, $\xi_{\rm E}^{(1)}$, and couples it to the signal in the channel through an arbitrary unitary operation, $\hat{U}^{(1)}_{\rm BE}$. At the end of the round, based on $\hat{U}^{(1)}_{\rm BE}$, $\xi_{\rm E}^{(1)}$, and the outcome of a potential measurement on $E$, the classical register $R$ updates from an initial state $R_{0}$ to a new state $R_{1}$, which determines (i) Eve's choice for the next unitary operation, $\hat{U}^{(2)}_{\rm BE}$, and (ii) how the state of the probe system $E$ is updated to a new state $\xi_{\rm E}^{(2)}$. Subsequent rounds proceed in the same way: at the end of round $u$, the state of the register, $R_{u}$, fully determines the next intervention of Eve, generally influenced by the record of all unitaries, $\{\hat{U}^{(v)}_{\rm BE}\}_{v\leq{}u}$, probe states, $\{\xi_{\rm E}^{(v)}\}_{v\leq{}u}$, and measurement outcomes.

Importantly, a consequence of this model is that earlier rounds can only affect the $u$-th round through their influence on $R_{u-1}$. From a technical point of view, the model is invoked to assure that the trace distance (TD) argument of Appendix~\ref{TD_argument} is applicable to the conditional detection statistics of Kato's inequality~\cite{Kato} (a similar adversary model is assumed, for instance, in~\cite{Margarida,Navs,Metger}). This said, we remark that, if each protocol round starts only after the conclusion of the previous one ---a sequential structure that can be practically enforced by Alice and Bob--- this adversary model is fully general.

\section{Characterization of the fully passive source}\label{Characterization}

The passive source we consider is due to~\cite{Zapatero} and shown in Fig.~\ref{fig:PT}. Although it uses four independent lasers, it can also be realized with an equivalent single-laser configuration and suitable interferometry, as discussed in~\cite{Mike,Zapatero} and implemented in~\cite{Lu,Hu}.

\begin{figure}[!htbp]
	\centering 
	\includegraphics[width=8.2cm,height=4.2cm]{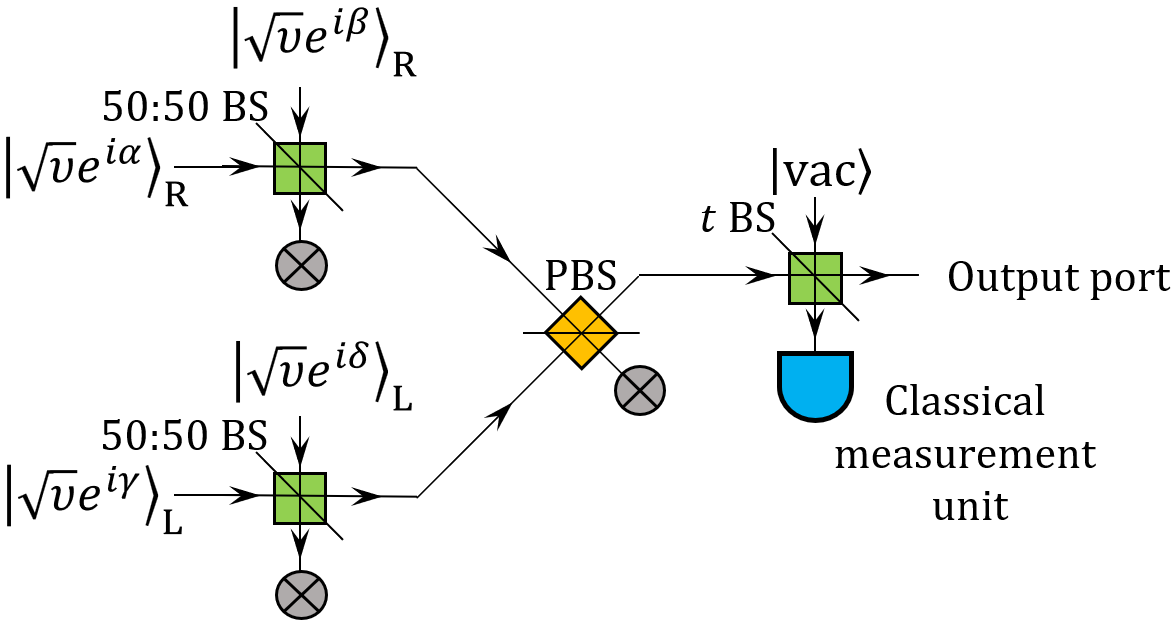}\\
	\caption{Possible architecture of a fully passive QKD source due to~\cite{Zapatero}. All coherent states in the figure have a common large intensity $\nu$ and independent random phases $\alpha$, $\beta$, $\gamma$ and $\delta$. These random phases can be achieved by operating the lasers under gain-switching conditions, \textit{i.e.}, by turning the lasers on and off between pulses. The interference in the 50:50 beam splitter (BS) of the top (bottom) arm yields a coherent state with random intensity~\cite{WCP_2} ---dependent on the phase difference $\beta-\alpha$ ($\delta-\gamma$)--- and right-handed (left-handed) circular polarization, R (L). When these two pulses interfere in the polarizing BS (PBS), a new coherent state is generated, whose intensity is the sum of the input intensities and whose polarization is randomly distributed in the RL sphere~\cite{passive_BB84}. Lastly, this final state enters a BS with transmittance $t\ll{1}$. The transmitted signal, which is attenuated to the single-photon level, is sent to Bob, and the reflected signal reaches a classical measurement unit for the accurate determination of the prepared intensity and polarization. Unused spatial modes are tagged by the symbol ``$\otimes$" in the figure.}
	\label{fig:PT}
\end{figure}

In the figure, $\ket{\tau}_{\rm R(L)}$ denotes a right-handed (left-handed) circularly polarized coherent state with complex amplitude $\tau\in\mathbb{C}$. That is to say, $\ket{\tau}_{\rm R(L)}=\exp\bigl\{\tau{}a_{\rm R(L)}^{\dagger}-\tau^{*}{}a_{\rm R(L)}\bigr\}\ket{\rm vac}$, where $\ket{\rm vac}$ is the vacuum state and $a^{\dagger}_{\rm R(L)}$ and $a_{\rm R(L)}$ respectively denote the creation and annihilation operators of a right-handed (left-handed) circular polarization mode. Throughout this work, we shall refer to the Bloch sphere spanned by $\ket{\rm R}=a^{\dagger}_{\rm R}\ket{\rm vac}$ and $\ket{\rm L}=a^{\dagger}_{\rm L}\ket{\rm vac}$ as the RL sphere, an arbitrary state of which reads
\begin{equation}\label{creation_operator}
\cos\left(\frac{\theta}{2}\right)\ket{\rm R}+e^{i\phi}\sin\left(\frac{\theta}{2}\right)\ket{\rm L},
\end{equation}
for $\theta\in[0,\pi]$ (polar angle) and $\phi\in(-\pi,\pi]$ (azimuthal angle). As shown in~\cite{Zapatero}, the mixed output state of the fully passive source reads
\begin{equation}\label{sigma}
\sigma=\frac{1}{2\pi}\int_{-\pi}^{\pi}d\phi\int_{0}^{\pi}d\theta\int_{0}^{I^{*}_{\theta}}dI{}f(\theta,I)\sum_{n=0}^{\infty}\frac{e^{-I}{I}^{n}}{n!}\ketbra{n}{n}_{\theta,\phi},
\end{equation}
for
\begin{equation}\begin{split}\label{distribution}
&f(\theta,I)=\frac{1}{2\nu{}t{}\pi^{2}\sqrt{1-\frac{I}{2\nu{}t}\displaystyle{\cos^{2}\left(\frac{\theta}{2}\right)}}\sqrt{1-\frac{I}{2\nu{}t}\sin^{2}\displaystyle{\left(\frac{\theta}{2}\right)}}},\\
&\ket{n}_{\theta,\phi}=\frac{1}{n!}\left[\cos\left(\frac{\theta}{2}\right)a^{\dagger}_{\rm R}+e^{i\phi}\sin\left(\frac{\theta}{2}\right)a^{\dagger}_{\rm L}\right]^{n}\ket{\rm vac},
\end{split}\end{equation}
and $I^{*}_{\theta}=\min\left\{2\nu{}t/\cos^{2}\left(\theta/2\right),2\nu{}t/\sin^{2}\left(\theta/2\right)\right\}$, $\nu$ and $t$ being introduced in Fig.~\ref{fig:PT}. According to these equations, $\phi$ is uniformly distributed in the output state, but $\theta$ is coupled to the intensity $I$.

\section{Encoding scheme and entanglement-based picture}\label{post-selection}
We select the encoding scheme of~\cite{Mike} ---which post-selects polar (equatorial) regions of the Bloch sphere for key generation (parameter estimation)--- and the decoy-state scheme of~\cite{Zapatero} ---which uses decoy intervals in both bases---. According to our simulations, these choices lead to a better performance than, say, only post-selecting equatorial regions for the encoding, or using decoy states in the test basis alone.

To be precise, let $\Gamma^{\rm key}$ and $\Gamma^{\rm test}$ denote two arbitrary lists of settings. The post-selection regions are given by
\begin{equation}\begin{split}\label{key_regions}
&\Omega^{\rm R}_{j}=\biggl\{\phi\in(-\pi,\pi],\theta\in\left(0,\Delta\theta'\right),I\in{}I_{j}\biggr\},\\
&\Omega^{\rm L}_{j}=\biggl\{\phi\in(-\pi,\pi],\theta\in\left(\pi-\Delta\theta',\pi\right),I\in{}I_{j}\biggr\}
\end{split}\end{equation}
with $j\in\Gamma^{\rm key}$ for key generation, and
\begin{equation}\begin{split}\label{test_regions}
&\Omega^{\rm H}_{j}=\\
&\biggl\{\phi\in\left(-\Delta\phi,\Delta\phi\right),\theta\in\left(\frac{\pi}{2}-\Delta\theta,\frac{\pi}{2}+\Delta\theta\right),I\in{}I_{j}\biggr\},\\
&\Omega^{\rm V}_{j}=\\
&\biggl\{\phi\in\left(\pi-\Delta\phi,\pi+\Delta\phi\right),\theta\in\left(\frac{\pi}{2}-\Delta\theta,\frac{\pi}{2}+\Delta\theta\right),I\in{}I_{j}\biggr\}
\end{split}\end{equation}
with $j\in\Gamma^{\rm test}$ for parameter estimation (PE). The polarization acceptance regions in the RL sphere are illustrated in Figure~\ref{fig:post-selection}, where the notation $\ket{\rm H}=(\ket{\rm R}+\ket{\rm L})/\sqrt{2}$ and $\ket{\rm V}=(\ket{\rm R}-\ket{\rm L})/\sqrt{2}$ is introduced.
\begin{figure}[H]
	\centering 
	\includegraphics[width=5.1cm,height=5.1cm]{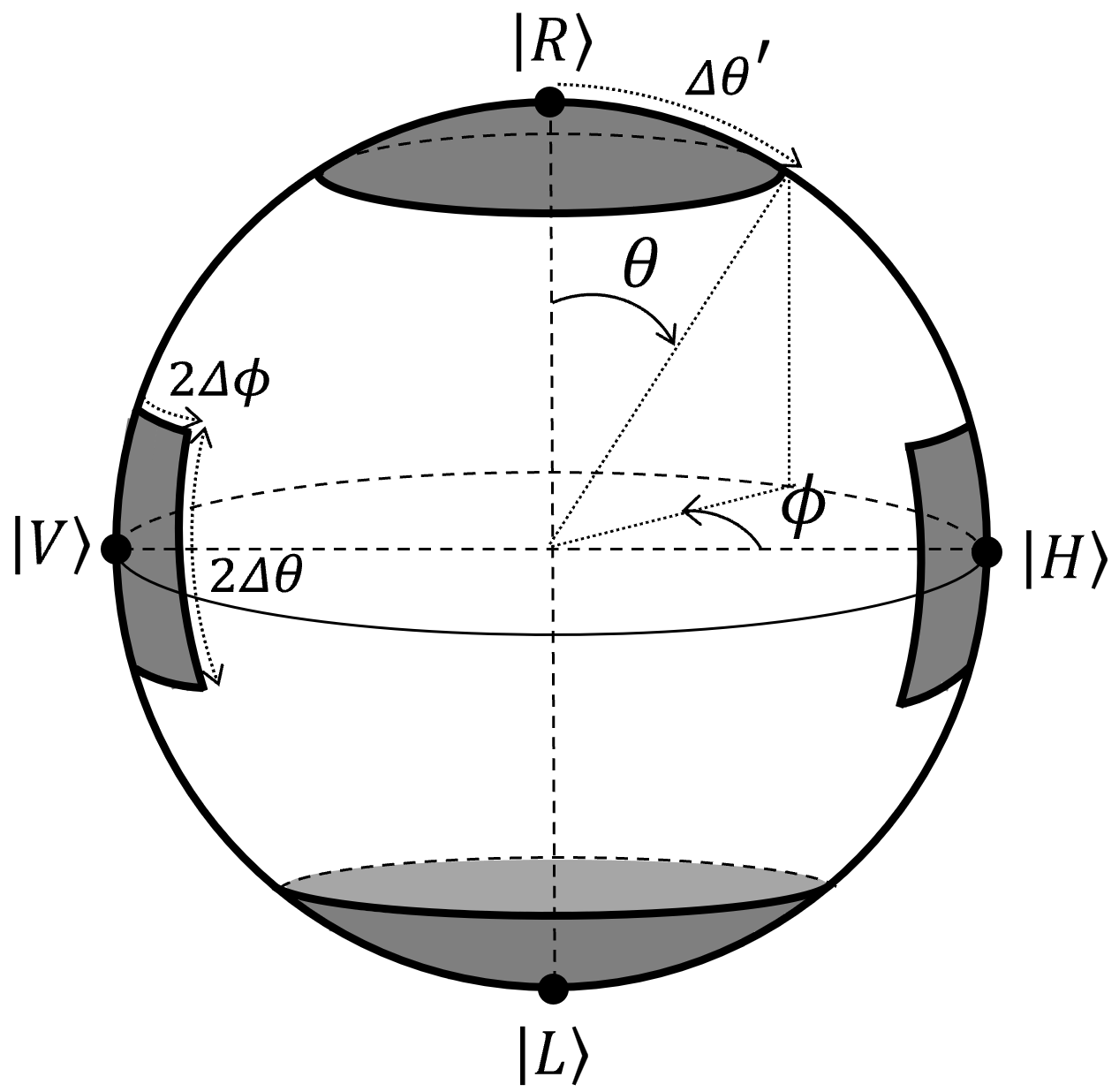}\\
	\caption{Post-selection of encoding regions in the RL sphere. While $\Delta\theta'$ characterizes the key regions (polar caps), $\Delta\theta$ and $\Delta\phi$ characterize the test regions (rectangular stripes).}
	\label{fig:post-selection}
\end{figure}
The post-selected state that arises from the acceptance region $\Omega^{\rm R(L)}_{j}$ can be written as
\begin{equation}\label{sigma_R}
\sigma^{\rm R(L)}_{j}=\frac{1}{\left\langle{1}\right\rangle_{\Omega^{\rm R(L)}_{j}}}{\left\langle\sum_{n=0}^{\infty}\frac{e^{-I}{I}^{n}}{n!}\ketbra{n}{n}_{\theta,\phi}\right\rangle_{\Omega^{\rm R(L)}_{j}}},
\end{equation}
where $\left\langle\cdot\right\rangle_{\Omega}$ denotes averaging with the weight $f(\theta,I)/2\pi$ over the region $\Omega$ of the $(\phi,\theta,I)$-space. Eq.~(\ref{sigma_R}) trivially decomposes into a classical mixture of Fock states,
\begin{equation}\label{sigma_R_2}
\sigma^{\rm R(L)}_{j}=\sum_{n=0}^{\infty}p_{n|j}^{\rm key}\sigma_{j,n}^{\rm R(L)},
\end{equation}
where we have introduced the shorthand notation
$p_{n|j}^{\rm key}=\left\langle{e^{-I}{I}^{n}/{n!}}\right\rangle_{\Omega^{\rm R}_{j}}/\langle{1}\rangle_{\Omega^{\rm R}_{j}}=\left\langle{e^{-I}{I}^{n}/{n!}}\right\rangle_{\Omega^{\rm L}_{j}}/\langle{1}\rangle_{\Omega^{\rm L}_{j}}$ and $\sigma_{j,n}^{\rm R(L)}=\bigl\langle{e^{-I}{I}^{n}/n!\ketbra{n}{n}_{\theta,\phi}}\bigr\rangle_{\Omega^{\rm R(L)}_{j}}\bigl/\bigl\langle{e^{-I}{I}^{n}/n!}\bigr\rangle_{\Omega^{\rm R(L)}_{j}}$. Naturally, the complete key-generation region of setting $j$, given by $\Omega_{j}^{\rm key}=\Omega^{\rm R}_{j}\cup{}\Omega^{\rm L}_{j}$, yields a state of the form $\sigma_{j}^{\rm key}=\sum_{n=0}^{\infty}p_{n|j}^{\rm key}\sigma_{j,n}^{\rm key}$, where
\begin{equation}\label{sigma_sn}
\sigma_{j,n}^{\rm key}=\frac{\sigma_{j,n}^{\rm R}+\sigma_{j,n}^{\rm L}}{2}.
\end{equation}
As for the test regions, we define $\Omega_{j}^{\rm test}=\Omega^{\rm H}_{j}\cup{}\Omega^{\rm V}_{j}$ and denote the post-selected states that arise from $\Omega^{\rm H}_{j}$, $\Omega^{\rm V}_{j}$ and $\Omega_{j}^{\rm test}$ respectively as $\sigma^{\rm H}_{j}$, $\sigma^{\rm V}_{j}$ and $\sigma_{j}^{\rm test}$, which satisfy completely analogous relations to those of Eqs.~(\ref{sigma_R}),~(\ref{sigma_R_2}) and~(\ref{sigma_sn}). In particular, the test-basis Fock states read $\sigma_{j,n}^{\rm H(V)}=\bigl\langle{e^{-I}{I}^{n}/n!\ketbra{n}{n}_{\theta,\phi}}\bigr\rangle_{\Omega^{\rm H(V)}_{j}}\bigl/\bigl\langle{e^{-I}{I}^{n}/n!}\bigr\rangle_{\Omega^{\rm H(V)}_{j}}$, and their corresponding photon-number statistics are $p_{n|j}^{\rm test}=\left\langle{e^{-I}{I}^{n}/{n!}}\right\rangle_{\Omega^{\rm H}_{j}}/\langle{1}\rangle_{\Omega^{\rm H}_{j}}=\left\langle{e^{-I}{I}^{n}/{n!}}\right\rangle_{\Omega^{\rm V}_{j}}/\langle{1}\rangle_{\Omega^{\rm V}_{j}}$.

Let us now focus on the single-photon components of $\sigma_{j}^{\rm R}$ and $\sigma_{j}^{\rm L}$, from which the final key is ultimately extracted. Using the matrix representation of Appendix~\ref{representation}, it follows that
\begin{equation}\begin{split}\label{depolarizing}
&\sigma_{j,1}^{\rm R}=\frac{1+\lambda_{j}^{\rm key}}{2}\ketbra{\rm R}{\rm R}+\frac{1-\lambda_{j}^{\rm key}}{2}\ketbra{\rm L}{\rm L},\\
&\sigma_{j,1}^{\rm L}=\frac{1+\lambda_{j}^{\rm key}}{2}\ketbra{\rm L}{\rm L}+\frac{1-\lambda_{j}^{\rm key}}{2}\ketbra{\rm R}{\rm R},
\end{split}\end{equation}
for $\lambda_{j}^{\rm key}=\left\langle{e^{-I}{I}\cos{\theta}}\right\rangle_{\Omega^{\rm R}_{j}}\bigl/\left\langle{e^{-I}{I}}\right\rangle_{\Omega^{\rm R}_{j}}$, meaning that the desired polarization is prepared with probability $(1+\lambda_{j}^{\rm key})/2$, and the orthogonal polarization is prepared otherwise. Without loss of generality then, one can describe these states by considering a shield system $S$ that stores the bit flip information. This leads to the purified states
\begin{equation}\begin{split}\label{shield}
&\ket{\Psi_{j,1}^{\rm R}}_{\rm SB}=\sqrt{\frac{1+\lambda_{j}^{\rm key}}{2}}\ket{0}_{\rm S}\ket{\rm R}_{\rm B}+\sqrt{\frac{1-\lambda_{j}^{\rm key}}{2}}\ket{1}_{\rm S}\ket{\rm L}_{\rm B},\\
&\ket{\Psi_{j,1}^{\rm L}}_{\rm SB}=\sqrt{\frac{1+\lambda_{j}^{\rm key}}{2}}\ket{0}_{\rm S}\ket{\rm L}_{\rm B}+\sqrt{\frac{1-\lambda_{j}^{\rm key}}{2}}\ket{1}_{\rm S}\ket{\rm R}_{\rm B},
\end{split}\end{equation}
where we have added the subscripts S and B for clarity (nonetheless, they shall be omitted where possible). Further proceeding with the standard entanglement-based source-replacement scheme at Alice's side, with $\{\ket{\rm R}_{\rm A},\ket{\rm L}_{\rm A}\}$ denoting an orthonormal basis of Alice's ancillary qubit $A$, we obtain the state
\begin{equation}\label{EB_key}
\ket{\Psi_{j,1}}_{\rm ASB}=\frac{1}{\sqrt{2}}\left(\ket{\rm R}_{\rm A}\ket{\Psi_{j,1}^{\rm R}}_{\rm SB}+\ket{\rm L}_{\rm A}\ket{\Psi_{j,1}^{\rm L}}_{\rm SB}\right),
\end{equation}
which equivalently describes the state preparation every time a single-photon is generated and $\Omega_{j}^{\rm key}$ is post-selected. Importantly, by introducing $\ket{\rm H}_{\rm A}=(\ket{\rm R}_{\rm A}+\ket{\rm L}_{\rm A})/\sqrt{2}$ and $\ket{\rm V}_{\rm A}=(\ket{\rm R}_{\rm A}-\ket{\rm L}_{\rm A})/\sqrt{2}$ and regrouping terms, $\ket{\Psi_{j,1}}_{\rm ASB}$ can be written as
\begin{equation}
\ket{\Psi_{j,1}}_{\rm ASB}=\frac{1}{\sqrt{2}}\left(\ket{\rm H}_{\rm A}\ket{\beta_{j}^{+}}_{\rm S}\ket{\rm H}_{\rm B}+\ket{\rm V}_{\rm A}\ket{\beta_{j}^{-}}_{\rm S}\ket{\rm V}_{\rm B}\right),
\end{equation}
where the shield states are given by $\ket{\beta_{j}^{\pm}}_{\rm S}=[(1+\lambda_{j}^{\rm key})/2]^{1/2}\ket{0}_{\rm S}\pm\bigl[(1-\lambda_{j}^{\rm key})/2\bigr]^{1/2}\ket{1}_{\rm S}$.

In summary, if Alice measures her ancilla A of $\ket{\Psi_{j,1}}_{\rm ASB}$ in the key basis $\{\ket{\rm R}_{\rm A},\ket{\rm L}_{\rm A}\}$, she prepares $\sigma_{j,1}^{\rm R}$ or $\sigma_{j,1}^{\rm L}$ at random, while if she measures it in the test basis $\{\ket{\rm H}_{\rm A},\ket{\rm V}_{\rm A}\}$, she prepares $\ket{\rm H}_{\rm B}$ or $\ket{\rm V}_{\rm B}$ at random irrespectively of $j\in\Gamma^{\rm key}$. As a consequence, for all $j\in\Gamma^{\rm key}$, the phase-error rate (PHER) of $\sigma_{j,1}^{\rm R}$ and $\sigma_{j,1}^{\rm L}$ could be estimated from the bit-error rate (BER) of the pure states $\ket{\rm H}_{\rm B}$ and $\ket{\rm V}_{\rm B}$ (the reader is referred to Appendix~\ref{security_claims} for a formal definition of the PHER).
Indeed, by post-selecting the test regions $\Omega_{j}^{\rm H}$ and $\Omega_{j}^{\rm V}$ of Eq.~(\ref{test_regions}), such perfect states are actually prepared at a certain rate, because the single-photon components of $\sigma_{j}^{\rm H}$ and $\sigma_{j}^{\rm V}$ verify
\begin{equation}\begin{split}\label{depolarizing_test}
&\sigma_{j,1}^{\rm H}=\lambda_{j}^{\rm test}\ketbra{\rm H}{\rm H}+(1-\lambda_{j}^{\rm test})\frac{\mathds{1}}{2},\\
&\sigma_{j,1}^{\rm V}=\lambda_{j}^{\rm test}\ketbra{\rm V}{\rm V}+(1-\lambda_{j}^{\rm test})\frac{\mathds{1}}{2},
\end{split}\end{equation}
for $\lambda_{j}^{\rm test}=\langle{e^{-I}{I}\sin{\theta}\cos{\phi}}\rangle_{\Omega^{\rm H}_{j}}\bigl/\langle{e^{-I}{I}}\rangle_{\Omega^{\rm H}_{j}}$ and all $j\in\Gamma^{\rm test}$. In other words, the post-selection of $\sigma_{j,1}^{\rm H}$ ($\sigma_{j,1}^{\rm V}$) determines the preparation of the ideal state $\ket{\rm H}$ ($\ket{\rm V}$) with probability $\lambda_{j}^{\rm test}$, and the fully mixed state with probability $1-\lambda_{j}^{\rm test}$. Exploiting this feature, in Section~\ref{decoy} we provide a decoy-state method that targets the BER of the ideal components $\ket{\rm H}_{\rm B}$ and $\ket{\rm V}_{\rm B}$ of $\sigma_{j,1}^{\rm H}$ and $\sigma_{j,1}^{\rm V}$ directly, rather than the mixed-state BER contemplated in the asymptotic analyses~\cite{Mike,Zapatero}.

\section{Protocol description}\label{Protocol_description}
For ease of understanding, we provide a description of a fully passive decoy-state BB84 protocol before explaining our PE method.

\begin{framed}
\hspace{.5cm}For a pre-agreed number of rounds $N$, the parties do the following.

\begin{enumerate}[label=(\roman*)]
\item{\it State preparation.} Alice operates the passive source and verifies if the delivered state matches any acceptance region.
\item{\it Measurement.} Bob measures the incoming signal in the key (test) basis with probability $q_{\mathds{K}}$ ($q_{\mathds{T}}=1-q_{\mathds{K}}$).\\

After the quantum communication phase, the public discussion and the classical post-processing run as follows.
  
\item{\it Public discussion.} Alice communicates Bob the set of rounds where $\sigma_{j}^{\rm key}$ ($\sigma_{j}^{\rm test}$) was post-selected, for all $j\in\Gamma^{\rm key}$ ($j\in\Gamma^{\rm test}$). In return, Bob reveals Alice the subsets of these sets where a basis match occurred and a detection event was recorded, say $\{\mathcal{X}_{j}^{\rm key}\}_{j\in\Gamma^{\rm key}}$ and $\{\mathcal{X}_{j}^{\rm test}\}_{j\in\Gamma^{\rm test}}$. $\mathcal{X}^{\rm key}=\bigcup_{j\in\Gamma^{\rm key}}\mathcal{X}_{j}^{\rm key}$ ($\mathcal{X}^{\rm test}=\bigcup_{j\in\Gamma^{\rm test}}\mathcal{X}_{j}^{\rm test}$) identifies the sifted-key data (test-basis PE data). We refer to $M_{j}^{\rm key}=\lvert{\mathcal{X}_{j}^{\rm key}}\rvert$ and $M_{j}^{\rm test}=\lvert{\mathcal{X}_{j}^{\rm test}}\rvert$ as the ``numbers of observed counts". Lastly, Alice discloses her bit strings for the test sets $\mathcal{X}_{j}^{\rm test}$.
\item{\it Parameter estimation.} Comparing Alice's test bit strings with his own, Bob computes the corresponding ``numbers of error counts", say $\{m_{j}^{\mathrm{test}}\}_{j\in\Gamma^{\rm test}}$. From his available data, Bob estimates a lower bound $M^{\rm L}_{\rm key,1}$ on the number of single-photon counts in $\mathcal{X}^{\rm key}$, $M_{\rm key,1}$, and an upper bound $e^{\rm (ph)\hspace{.05cm}\rm U}_{1}$ on the single-photon PHER in this set, $e^{\rm (ph)}_{1}$. Bob uses the above bounds to compute the secret key length, $l$, for a pre-fixed error correction (EC) leakage, $\lambda_{\rm EC}$, and a pre-fixed error verification (EV) tag size, $\lceil{\log(1/\epsilon_{\rm cor})}\rceil$. If $l=0$, Bob aborts the protocol.
\item{\it Error correction and error verification.}  Upon success of the PE step, the parties run the pre-agreed EC and EV steps, the latter being based on two-universal hashing~\cite{Wegman}. If EV fails, Bob aborts the protocol.
\item{\it Privacy amplification.}  Upon success of the EV step, the parties apply privacy amplification (PA) based on two-universal hashing~\cite{Wegman,LHL} to their reconciled sifted keys, obtaining final keys of length $l$.
\end{enumerate}
\end{framed}
The reader is referred to Appendix~\ref{security_claims} for a derivation of the secret key length $l$.
\section{Decoy-state method}\label{decoy}
Here, we generalize the decoy-state method of~\cite{Zapatero} to the finite-key regime, incorporating a novelty that surpasses a limitation of both~\cite{Mike,Zapatero}.

We recall that, following the adversary model of Section~\ref{Assumptions}, the state of Eve's register at the end of round $u-1$, $R_{u-1}$, fully determines Eve's intervention in round $u$. This said, the notation goes as follows. For $u=1,2\ldots{}N$, $Q^{\beta\hspace{.05cm}(u)}_{j}$ ($E^{\beta\hspace{.05cm}(u)}_{j}$) stands for the probability that a ``click" (an ``error") is recorded in round $u$, conditioned on $R_{u-1}$ and on the event that $\sigma_{j}^{\beta}$ is post-selected and Bob performs a $\beta$-basis measurement, $\beta\in\{\mathrm{key,test}\}$. That is to say, $Q^{\beta\hspace{.05cm}(u)}_{j}=p^{(u)}\bigl(\mathrm{click}|R_{u-1},\sigma_{j}^{\beta},\beta\bigr)$ and $E^{\beta\hspace{.05cm}(u)}_{j}=p^{(u)}\bigl(\mathrm{error}|R_{u-1},\sigma_{j}^{\beta},\beta\bigr)$. Similarly, $y^{\beta\hspace{.05cm}(u)}_{j,n}$ and $e^{\beta\hspace{.05cm}(u)}_{j,n}$ denote the corresponding $n$-photon yield and $n$-photon error probability, $y^{\beta\hspace{.05cm}(u)}_{j,n}=p^{(u)}\bigl(\mathrm{click}|R_{u-1},\sigma^{\beta}_{j,n},\beta\bigr)$ and $e^{\beta\hspace{.05cm}(u)}_{j,n}=p^{(u)}\bigl(\mathrm{error}|R_{u-1},\sigma^{\beta}_{j,n},\beta\bigr)$. Note that ``error" refers here to the joint event where a click is recorded and it triggers a bit error.
\subsection{Decoy constraints}
From the photon-number decomposition of $\sigma^{\beta}_{j}$, it follows that $Q^{\beta\hspace{.05cm}(u)}_{j}=\sum_{n=0}^{\infty}p^{\beta}_{n|j}y^{\beta\hspace{.05cm}(u)}_{j,n}$ and $E^{\beta\hspace{.05cm}(u)}_{j}=\sum_{n=0}^{\infty}p^{\beta}_{n|j}e^{\beta\hspace{.05cm}(u)}_{j,n}$ for all $u=1,2\ldots{}N$, $j\in\Gamma^{\beta}$ and $\beta\in\{\mathrm{key,test}\}$. If we now define the averaged quantities $Q^{\beta}_{j}=\sum_{u}Q^{\beta\hspace{.05cm}(u)}_{j}/N$, $E^{\beta}_{j}=\sum_{u}E^{\beta\hspace{.05cm}(u)}_{j}/N$, $y^{\beta}_{j,n}=\sum_{u}y^{\beta\hspace{.05cm}(u)}_{j,n}/N$ and $e^{\beta}_{j,n}=\sum_{u}e^{\beta\hspace{.05cm}(u)}_{j,n}/N$, and truncate the index $n$ to a threshold photon number $n_{\rm cut}$, we derive the sets of constraints
\begin{equation}\label{decoy_constraints}
\begin{cases}
Q^{\beta}_{j}\geq\displaystyle\sum_{n=0}^{n_{\rm cut}}p^{\beta}_{n|j}y^{\beta}_{j,n}, \\
Q^{\beta}_{j}\leq\displaystyle\sum_{n=0}^{n_{\rm cut}}p^{\beta}_{n|j}y^{\beta}_{j,n}+1-\displaystyle\sum_{n=0}^{n_{\rm cut}}p^{\beta}_{n|j},  
\end{cases}
\end{equation}
and
\begin{equation}\label{decoy_constraints_2}
\begin{cases}
E^{\beta}_{j}\geq\displaystyle\sum_{n=0}^{n_{\rm cut}}p^{\beta}_{n|j}e^{\beta}_{j,n}, \\
E^{\beta}_{j}\leq\displaystyle\sum_{n=0}^{n_{\rm cut}}p^{\beta}_{n|j}e^{\beta}_{j,n}+1-\displaystyle\sum_{n=0}^{n_{\rm cut}}p^{\beta}_{n|j},
\end{cases}
\end{equation}
to be addressed numerically later on.
\subsection{Trace-distance constraints}
Coming next, we reproduce the TD constraints provided in~\cite{Zapatero}, to be combined with the decoy constraints above. For each basis $\beta$ and photon number $n$, these extra constraints restrict the differences $\abs*{y^{\beta}_{j,n}-y^{\beta}_{k,n}}$ and $\abs*{e^{\beta}_{j,n}-e^{\beta}_{k,n}}$ for every pair of distinct settings $(j,k)$.

\subsubsection{Yields}
Regarding the yields, Bob's possible measurement outcomes are ``click" and ``no click", such that his measurement can be described by a positive-operator-valued measure (POVM) with elements $\{\hat{M}^{\mathrm{click}}_{\rm B},\hat{M}^{\mathrm{no}\hspace{.05cm}\mathrm{click}}_{\rm B}\}$, where $\hat{M}^{\mathrm{no}\hspace{.05cm}\mathrm{click}}_{\rm B}=\mathds{1}_{\rm B}-\hat{M}^{\mathrm{click}}_{\rm B}$ (note that this POVM is basis-independent due to the basis-independent detection efficiency assumption). Therefore, from the adversary model described in Section~\ref{Assumptions}, we have that
\begin{equation}
y^{\beta\hspace{.05cm}(u)}_{j,n}=\Tr\left[\hat{D}^{\mathrm{click}\hspace{.05cm}(u)}_{\rm BE}\ \sigma_{j,n}^{\beta}\otimes\xi_{\rm E}^{(u)}\right]
\end{equation}
for $\hat{D}^{\mathrm{click}\hspace{.05cm}(u)}_{\rm BE}=\hat{U}^{(u)\dagger}_{\rm BE}\hat{M}^{\mathrm{click}}_{\rm B}\hat{U}^{(u)}_{\rm BE}$, where $\xi_{\rm E}^{(u)}$ and $\hat{U}^{(u)}_{\rm BE}$ respectively denote the state of Eve's quantum probe and Eve's unitary operation in round $u$.

Now, in virtue of the TD argument presented in Appendix~\ref{TD_argument}, it follows that
\begin{equation}\label{TD_yield}
\abs{y^{\beta\hspace{.05cm}(u)}_{j,n}-y^{\beta\hspace{.05cm}(u)}_{k,n}}\leq{}D\left(\sigma^{\beta}_{j,n},\sigma^{\beta}_{k,n}\right),
\end{equation}
for all possible inputs $\beta$, $j$, $k$, $n$ and $u$, where $D(\rho,\tau)=\frac{1}{2}\Tr[\sqrt{(\rho-\tau)^{2}}]$ denotes the trace distance between $\rho$ and $\tau$. In particular, we remark that $\sigma^{\rm key}_{j,1}=\mathds{1}/2$ ($\sigma^{\rm key}_{j,0}=\ketbra{\rm vac}{\rm vac}$) for all $j\in\Gamma^{\rm key}$ and $\sigma^{\rm test}_{j,1}=\mathds{1}/2$ ($\sigma^{\rm test}_{j,0}=\ketbra{\rm vac}{\rm vac}$) for all $j\in\Gamma^{\rm test}$ ---see Eq.~(\ref{depolarizing}) and Eq.~(\ref{depolarizing_test})---. As a consequence, the single-photon yield and the vacuum yield are independent of both the intensity setting and the basis.
\subsubsection{Test-basis error probabilities}
In order to describe the test-basis error probabilities, finer-grained measurement operators are required, in a one-to-one correspondence with Bob's possible outcomes: ``H", ``V" and ``no click". As usual, double clicks are randomly assigned to either ``H" or ``V", and this assignment is straightforwardly incorporated in the POVM. In short, error-wise, Bob's measurement is described by a POVM with elements $\{\hat{M}^{\rm H}_{\rm B},\hat{M}^{\rm V}_{\rm B},\hat{M}^{\mathrm{no}\hspace{.05cm}\mathrm{click}}_{\rm B}\}$, where $\hat{M}^{\rm H}_{\rm B}+\hat{M}^{\rm V}_{\rm B}=\hat{M}^{\mathrm{click}}_{\rm B}$. Moreover, since $\sigma^{\rm test}_{j,n}=\bigl(\sigma_{j,n}^{\rm H}+\sigma_{j,n}^{\rm V}\bigr)/2$ it follows that
\begin{equation}\begin{split}\label{error_yield_def}
&e_{j,n}^{\mathrm{test}\hspace{.05cm}(u)}=\\
&\frac{1}{2}\bigl[p^{(u)}\left(\mathrm{V}|R_{u-1},\sigma_{j,n}^{\rm H},\mathrm{test}\right)+p^{(u)}\left(\mathrm{H}|R_{u-1},\sigma_{j,n}^{\rm V},\mathrm{test}\right)\bigr]=\\
&\frac{1}{2}\biggl\{\Tr\left[\hat{D}^{\mathrm{V}\hspace{.05cm}(u)}_{\rm BE}\ \sigma_{j,n}^{\rm H}\otimes\xi_{\rm E}^{(u)}\right]+\Tr\left[\hat{D}^{\mathrm{H}\hspace{.05cm}(u)}_{\rm BE}\ \sigma_{j,n}^{\rm V}\otimes\xi_{\rm E}^{(u)}\right]\biggr\},
\end{split}\end{equation}
for $\hat{D}^{y\hspace{.05cm}(u)}=\hat{U}^{(u)\dagger}_{\rm BE}\hat{M}^{y}_{\rm B}\hat{U}^{(u)}_{\rm BE}$ with $y\in\{\rm H,\rm V\}$. Coming next, combining the TD argument with the triangle inequality, one can readily show that
\begin{equation}
\left\lvert{e_{j,n}^{\mathrm{test}\hspace{.05cm}(u)}-e_{k,n}^{\mathrm{test}\hspace{.05cm}(u)}}\right\rvert\leq{}\frac{1}{2}\left[D\left(\sigma_{j,n}^{\rm H},\sigma_{k,n}^{\rm H}\right)+D\left(\sigma_{j,n}^{\rm V},\sigma_{k,n}^{\rm V}\right)\right].
\end{equation}
Moreover, because of the azimuthal symmetry of the output state of the source, $D\bigl(\sigma_{j,n}^{\rm H},\sigma_{k,n}^{\rm H}\bigr)=D\bigl(\sigma_{j,n}^{\rm V},\sigma_{k,n}^{\rm V}\bigr)$, and thus
\begin{equation}\label{TD_error}
\left\lvert{e_{j,n}^{\mathrm{test}\hspace{.05cm}(u)}-e_{k,n}^{\mathrm{test}\hspace{.05cm}(u)}}\right\rvert\leq{}D\left(\sigma_{j,n}^{\rm H},\sigma_{k,n}^{\rm H}\right).
\end{equation}
In particular, $e_{j,0}^{\mathrm{test}\hspace{.05cm}(u)}=e_{k,0}^{\mathrm{test}\hspace{.05cm}(u)}$ for all $j$ and $k$ but $e_{j,1}^{\mathrm{test}\hspace{.05cm}(u)}\neq{}e_{k,1}^{\mathrm{test}\hspace{.05cm}(u)}$ because $\sigma_{j,1}^{\rm H}\neq{}\sigma_{k,1}^{\rm H}$ if $j\neq{}k$.

Since the TD bounds in Eq.~(\ref{TD_yield}) and Eq.~(\ref{TD_error}) are round-independent, the exact same bounds hold after averaging over all protocol rounds (this is an immediate consequence of the triangle inequality too). That is to say, $\abs*{y^{\beta}_{j,n}-y^{\beta}_{k,n}}\leq{}D\bigl(\sigma^{\beta}_{j,n},\sigma^{\beta}_{k,n}\bigr)$ for $\beta\in\{\mathrm{key,test}\}$, and $\abs*{e^{\rm test}_{j,n}-e^{\rm test}_{k,n}}\leq{}D\bigl(\sigma_{j,n}^{\rm H},\sigma_{k,n}^{\rm H}\bigr)$ for all possible $j,k$ and $n$. In order to explicitly calculate these TD values, we use the matrix representation described in Appendix~\ref{representation}.

\subsection{Noise-suppressing constraints}\label{noise_suppression}
As shown in Sec.~\ref{post-selection}, the BER of the ideal states $\ket{\rm H}$ and $\ket{\rm V}$ is the relevant quantity to estimate the PHER of the key-generating states $\sigma^{\rm R}_{j,1}$ and $\sigma^{\rm L}_{j,1}$, and one can target this quantity directly in the LP by harnessing Eq.~(\ref{depolarizing_test}). This is what we do next.

By plugging Eq.~(\ref{depolarizing_test}) into Eq.~(\ref{error_yield_def}) and using (i) $\hat{D}^{\mathrm{H}\hspace{.05cm}(u)}_{\rm BE}+\hat{D}^{\mathrm{V}\hspace{.05cm}(u)}_{\rm BE}=\hat{D}^{\mathrm{click}\hspace{.05cm}(u)}_{\rm BE}$ and (ii) $\mathds{1}/{2}=\sigma_{j,1}^{\rm test}$, one obtains
\begin{equation}\begin{split}\label{BER_decomposition}
&e_{j,1}^{\mathrm{test}\hspace{.05cm}(u)}=\\
&\lambda_{j}^{\rm test}\frac{\Tr\left[\hat{D}^{\mathrm{V}\hspace{.05cm}(u)}_{\rm BE}\hspace{.025cm}\ketbra{\rm H}{\rm H}\otimes\xi_{\rm E}^{(u)}\right]+\Tr\left[\hat{D}^{\mathrm{H}\hspace{.05cm}(u)}_{\rm BE}\hspace{.025cm}\ketbra{\rm V}{\rm V}\otimes\xi_{\rm E}^{(u)}\right]}{2}\\
&+(1-\lambda_{j}^{\rm test})\frac{\Tr\left[\hat{D}^{\mathrm{V}\hspace{.05cm}(u)}_{\rm BE}\hspace{.025cm}\displaystyle{\frac{\mathds{1}}{2}}\otimes\xi_{\rm E}^{(u)}\right]+\Tr\left[\hat{D}^{\mathrm{H}\hspace{.05cm}(u)}_{\rm BE}\hspace{.025cm}\displaystyle{\frac{\mathds{1}}{2}}\otimes\xi_{\rm E}^{(u)}\right]}{2}=\\
&\lambda_{j}^{\rm test}\frac{p^{(u)}\left(\mathrm{V}|R_{u-1},\ketbra{\rm H}{\rm H},\mathrm{test}\right)+p^{(u)}\left(\mathrm{H}|R_{u-1},\ketbra{\rm V}{\rm V},\mathrm{test}\right)}{2}\\
&+(1-\lambda_{j}^{\rm test})\frac{p^{(u)}\left(\mathrm{click}|R_{u-1},\sigma^{\rm test}_{j,1},\mathrm{test}\right)}{2}.
\end{split}\end{equation}
By definition, the first term is the contribution of the ideal states $\ket{\rm H}$ and $\ket{\rm V}$ to the conditional bit-error probability in round $u$, say $e_{1}^{\mathrm{ideal}\hspace{.05cm}(u)}$, and the second term is the white-noise contribution from the fully mixed state. Noticing that $p^{(u)}\left(\mathrm{click}|R_{u-1},\sigma^{\rm test}_{j,1},\mathrm{test}\right)=y^{\mathrm{test}\hspace{.05cm}(u)}_{j,1}$, Eq.~(\ref{BER_decomposition}) reads $e_{j,1}^{\mathrm{test}\hspace{.05cm}(u)}=\lambda_{j}^{\rm test}e_{1}^{\mathrm{ideal}\hspace{.05cm}(u)}+(1-\lambda_{j}^{\rm test})y^{\mathrm{test}\hspace{.05cm}(u)}_{j,1}/2$, which leads to
\begin{equation}
e_{j,1}^{\mathrm{test}}=\lambda_{j}^{\rm test}e_{1}^{\mathrm{ideal}}+(1-\lambda_{j}^{\rm test})\frac{y^{\mathrm{test}}_{j,1}}{2}
\end{equation}
for the averaged quantities $e^{\rm test}_{j,1}=\sum_{u}e_{j,1}^{\mathrm{test}\hspace{.05cm}(u)}/N$, $e^{\rm ideal}_{1}=\sum_{u}e_{1}^{\mathrm{ideal}\hspace{.05cm}(u)}/N$ and $y^{\rm test}_{j,1}=\sum_{u}y_{j,1}^{\mathrm{test}\hspace{.05cm}(u)}/N$. We remark that the white-noise component $y^{\mathrm{test}}_{j,1}/2$ is setting-independent (just like the ideal component $e_{1}^{\rm ideal}$) due to the setting-independence of $y^{\mathrm{test}}_{j,1}$.
\subsection{Linear programs}\label{LPs}
Here, we select an arbitrary reference setting $\alpha\in\Gamma^{\rm key}$, such that $y_{\alpha,1}^{\rm key}$ becomes the target of the yield-related linear program (LP). Of course, in virtue of the setting-independence and the basis-independence of the single-photon yield, one could select a different setting from $\Gamma^{\rm key}$ or even from $\Gamma^{\rm test}$ for this purpose.

Combining the three types of constraints we have presented, one reaches final LPs for the estimation of the single-photon parameters $y^{\mathrm{key}}_{\alpha,1}$ and $e_{1}^{\mathrm{ideal}}$ given $\{Q_{j}^{\mathrm{key}}\}_{j\in\Gamma^{\rm key}}$ and $\{E_{j}^{\rm test}\}_{j\in\Gamma^{\rm test}}$. Note, however, that the latter cannot be determined with certainty in the protocol, but only estimated from the observed numbers of key-basis measure counts, $\{M_{j}^{\rm key}\}_{j\in\Gamma^{\rm key}}$, and test-basis error counts, $\{m_{j}^{\mathrm{test}}\}_{j\in\Gamma^{\rm test}}$, using concentration inequalities for sums of dependent random variables (RVs). In particular, we use Kato's inequality~\cite{Kato} for this purpose, which allows to establish that
\begin{equation}\begin{split}\label{Kato_LP}
&N{}q_{\mathds{K}}\left\langle{1}\right\rangle_{\Omega_{j}^{\rm key}}Q_{j}^{\mathrm{key}}\overset{2\epsilon}{\in}\biggl(K^{\rm L}_{N,\epsilon}(M_{j}^{\rm key}),K^{\rm U}_{N,\epsilon}(M_{j}^{\rm key})\biggr),\\
&N{}q_{\mathds{T}}\left\langle{1}\right\rangle_{\Omega_{j}^{\mathrm{test}}}E_{j}^{\mathrm{test}}\overset{2\epsilon}{\in}\biggl(K^{\rm L}_{N,\epsilon}(m_{j}^{\mathrm{test}}),K^{\rm U}_{N,\epsilon}(m_{j}^{\mathrm{test}})\biggr),
\end{split}\end{equation}
for known functions $K^{\rm L}_{N,\epsilon}(x)$ and $K^{\rm U}_{N,\epsilon}(x)$ given in Appendix~\ref{Kato}, where the superscripts $2\epsilon$ over the ``$\in$" symbols indicate that the corresponding intervals hold except with probability $2\epsilon$ at most~\cite{comment}. To be precise, in both intervals of Eq.~(\ref{Kato_LP}), each one-sided bound holds except with probability $\epsilon$ at most, which translates into an overall error probability of $2\epsilon$ in virtue of the union bound~\cite{union}. As an example, the reader is referred to Appendix~\ref{example} for the explicit derivation of one of the bounds in Eq.~(\ref{Kato_LP}).

In conclusion, we can take $\{Q_{j}^{\rm key}\}_{j\in\Gamma^{\rm key}}$ and $\{E_{j}^{\rm test}\}_{j\in\Gamma^{\rm test}}$ to be additional variables of the LPs, as long as we further incorporate the constraints set by Eq.~(\ref{Kato_LP})~\cite{detail}. Putting it all together, if one denotes $|\Gamma^{\rm key}|=d_{\rm key}$ and $|\Gamma^{\rm test}|=d_{\rm test}$, it follows that
\begin{equation}
y_{\alpha,1}^{\rm key}\overset{2\epsilon{}d_{\rm key}}{>}y_{1}^{\rm L}\hspace{.2cm}\mathrm{and}\hspace{.2cm}e_{1}^{\mathrm{ideal}}\overset{2\epsilon(d_{\rm key}+d_{\rm test})}{<}e_{1}^{\mathrm{ideal}\hspace{.05cm}\mathrm{U}}
\end{equation}
(the superscripts upper-bounding the error probabilities again), where $y_{1}^{\rm L}$ is the solution to
\begin{equation}\begin{split}\label{lp_1}
&\min\quad y_{\alpha,1}^{\rm key}\hspace{.4cm}\textup{s.t.}\\
&\sum_{n=0}^{n_{\rm cut}}p_{n|j}^{\rm key}y_{j,n}^{\rm key}\leq{}Q_{j}^{\rm key},\hspace{.1cm}j\in\Gamma^{\rm key},\\
&Q_{j}^{\rm key}\leq\sum_{n=0}^{n_{\rm cut}}p_{n|j}^{\rm key}y_{j,n}^{\rm key}+1-\sum_{n=0}^{n_{\rm cut}}p_{n|j}^{\rm key},\hspace{.1cm}j\in\Gamma^{\rm key},\\
&\left\lvert{y^{\rm key}_{j,n}-y^{\rm key}_{k,n}}\right\rvert\leq{}D\bigl(\sigma^{\rm key}_{j,n},\sigma^{\rm key}_{k,n}\bigr),\hspace{.1cm}j,k\in\Gamma^{\rm key},\ n=2\ldots{}n_{\rm cut},\\
&y^{\rm key}_{j,n}=y^{\rm key}_{k,n},\hspace{.1cm}j,k\in\Gamma^{\rm key},\ n=0,1,\\
&0\leq{}y^{\rm key}_{j,n}\leq{}1,\hspace{.1cm}j\in\Gamma^{\rm key},\ n=0\ldots{}n_{\rm cut},\\
&\frac{K^{\rm L}_{N,\epsilon}(M_{j}^{\rm key})}{N{}q_{\mathds{K}}\left\langle{1}\right\rangle_{\Omega_{j}^{\rm key}}}\leq{Q_{j}^{\rm key}}\leq{}\frac{K^{\rm U}_{N,\epsilon}(M_{j}^{\rm key})}{N{}q_{\mathds{K}}\left\langle{1}\right\rangle_{\Omega_{j}^{\rm key}}},\hspace{.1cm}j\in\Gamma^{\rm key},
\end{split}\end{equation}
and $e_{1}^{\mathrm{ideal}\hspace{.05cm}\mathrm{U}}$ is the solution to
\begin{equation}\begin{split}\label{lp_2}
&\max\quad e_{1}^{\mathrm{ideal}}\hspace{.4cm}\textup{s.t.}\\
&\sum_{n=0}^{n_{\rm cut}}p_{n|j}^{\rm test}e_{j,n}^{\rm test}\leq{}E_{j}^{\rm test},\hspace{.1cm}j\in\Gamma^{\rm test},\\
&E_{j}^{\rm test}\leq\sum_{n=0}^{n_{\rm cut}}p_{n|j}^{\rm test}e_{j,n}^{\rm test}+1-\sum_{n=0}^{n_{\rm cut}}p_{n|j}^{\rm test},\hspace{.1cm}j\in\Gamma^{\rm test},\\
&\left\lvert{e_{j,n}^{\rm test}-e_{k,n}^{\rm test}}\right\rvert\leq{}D\bigl(\sigma_{j,n}^{\rm H},\sigma_{k,n}^{\rm H}\bigr),\hspace{.1cm}j,k\in\Gamma^{\rm test},\ n=1\ldots{}n_{\rm cut},\\
&e^{\rm test}_{j,0}=e^{\rm test}_{k,0},\hspace{.1cm}j,k\in\Gamma^{\rm test},\\
&\lambda_{j}^{\rm test}e_{1}^{\mathrm{ideal}}\leq{}e_{j,1}^{\rm test}-(1-\lambda_{j}^{\rm test})\frac{y_{1}^{\rm L}}{2},\hspace{.1cm}j\in\Gamma^{\rm test},\\
&0\leq{}e_{1}^{\rm ideal}\leq{}1,\hspace{.1cm}0\leq{}e_{j,n}^{\rm test}\leq{}1,\hspace{.1cm}j\in\Gamma^{\rm test},\ n=0\ldots{}n_{\rm cut},\\
&\frac{K^{\rm L}_{N,\epsilon}(m_{j}^{\mathrm{test}})}{N{}q_{\mathds{T}}\left\langle{1}\right\rangle_{\Omega_{j}^{\rm test}}}\leq{E_{j}^{\rm test}}\leq{}\frac{K^{\rm U}_{N,\epsilon}(m_{j}^{\mathrm{test}})}{N{}q_{\mathds{T}}\left\langle{1}\right\rangle_{\Omega_{j}^{\rm test}}},\hspace{.1cm}j\in\Gamma^{\rm test}.
\end{split}\end{equation}
Note that the error probability of the first LP is bounded by $2\epsilon{}d_{\rm key}$ because it relies on $2\epsilon{}d_{\rm key}$ usages of Kato's inequality for the measure counts $\{M_{j}^{\rm key}\}_{j\in\Gamma^{\rm key}}$, each of them having a fixed error probability $\epsilon$. On the contrary, the second LP has an error probability of $2\epsilon(d_{\rm key}+d_{\rm test})$ because, besides relying on $2\epsilon{}d_{\rm test}$ usages of Kato's inequality for the error counts $\{m_{j}^{\mathrm{test}}\}_{j\in\Gamma^{\rm test}}$, it further uses the fact that, for all $j\in\Gamma^{\rm test}$, $y_{j,1}^{\rm test}=y_{\alpha,1}^{\rm key}\overset{2\epsilon{}d_{\rm key}}{>}y_{1}^{\rm L}$.

Alternatively, running the first LP under a maximization condition instead provides an upper bound $y_{1}^{\rm U}$ on $y_{\alpha,1}^{\rm key}$, such that $y_{\alpha,1}^{\rm key}\overset{2\epsilon{}d_{\rm key}}{<}y_{1}^{\rm U}$.
\section{Secret key parameters}\label{secret_key_parameters}
In order to evaluate the secret key length formula, given by Eq.~(\ref{keylength}), one must estimate the secret key parameters $M_{\rm key,1}$ and $e^{\rm (ph)}_{1}$. This is what we do next, using the decoy-state bounds $y_{1}^{\rm L}$, $y_{1}^{\rm U}$ and $e_{1}^{\mathrm{ideal}\hspace{.05cm}\mathrm{U}}$.

For simplicity, a common error probability $\epsilon$ shall be assumed for every usage of a concentration inequality in what follows, matching the error probability presumed for the Kato bounds in the previous section.
\subsection{Bounds on the single-photon measure counts}\label{yield}
Let us introduce the notation $\Omega^{\beta}=\bigcup_{j\in\Gamma^{\beta}}{\Omega_{j}^{\beta}}$ for $\beta\in\{\rm key,test\}$. From the reverse Kato bounds of Appendix~\ref{Kato} and an argument formally identical to that of Appendix~\ref{example} (but applied to a different sequence of Bernoulli RVs), one can show that
\begin{equation}\begin{split}\label{Kato_reverse}
M_{\rm key,1}\overset{2\epsilon}\in\Biggl(&\bar{K}^{\rm L}_{N,\epsilon}\Bigl(N{}q_{\mathds{K}}\left\langle{e^{-I}{I}}\right\rangle_{\Omega^{\rm key}}y_{\alpha,1}^{\rm key}\Bigr),\\
&\bar{K}^{\rm U}_{N,\epsilon}\Bigl(N{}q_{\mathds{K}}\left\langle{e^{-I}{I}}\right\rangle_{\Omega^{\rm key}}y_{\alpha,1}^{\rm key}\Bigr)\Biggr),
\end{split}\end{equation}
for known functions $\bar{K}^{\rm L}_{N,\epsilon}(x)$ and $\bar{K}^{\rm U}_{N,\epsilon}(x)$ given in Appendix~\ref{Kato}. As in Eq.~(\ref{Kato_LP}), each bound in the interval above holds except with probability $\epsilon$ at most. Also, we remark that Eq.~(\ref{Kato_reverse}) explicitly uses the fact that $y_{j,1}^{\rm key}=y_{\alpha,1}^{\rm key}$ for all $j\in\Gamma^{\rm key}$. Now, using $y_{\alpha,1}^{\rm key}\overset{2\epsilon{}d_{\rm key}}{>}y_{1}^{\rm L}$ and $y_{\alpha,1}^{\rm key}\overset{2\epsilon{}d_{\rm key}}{<}y_{1}^{\rm U}$ in Eq.~(\ref{Kato_reverse}), it follows that
\begin{equation}\begin{split}\label{Kato_reverse_2}
&M_{\rm key,1}\overset{\epsilon{}(2d_{\rm key}+1)}>\bar{K}^{\rm L}_{N,\epsilon}\Bigl(N{}q_{\mathds{K}}\left\langle{e^{-I}{I}}\right\rangle_{\Omega^{\rm key}}y_{1}^{\rm L}\Bigr)=:M_{\rm key,1}^{\rm L},\\
&M_{\rm key,1}\overset{\epsilon{}(2d_{\rm key}+1)}<\bar{K}^{\rm U}_{N,\epsilon}\Bigl(N{}q_{\mathds{K}}\left\langle{e^{-I}{I}}\right\rangle_{\Omega^{\rm key}}y_{1}^{\rm U}\Bigr)=:M_{\rm key,1}^{\rm U}.
\end{split}\end{equation}
In a similar fashion, one can derive a lower bound on the number of measure counts in $\mathcal{X}^{\rm test}$ triggered by the perfectly prepared test states $\ket{\rm H}$ and $\ket{\rm V}$, say $M_{\rm test,1}^{\rm ideal}$. This is so because $(\dyad{\rm H}{\rm H}+\dyad{\rm V}{\rm V})/2=\sigma_{\alpha,1}^{\rm key}=\mathds{1}/2$, and hence $y_{1}^{\rm L}$ provides a lower bound on the corresponding average yield too  ($y_{1}^{\rm L}$ lower bounds the round-averaged yield of any convex combination of single-photon states that adds up to $\mathds{1}/2$). Particularly, we have
\begin{equation}\begin{split}
&M_{\rm test,1}^{\rm ideal}\overset{\epsilon{}(2d_{\rm key}+1)}>\\
&\bar{K}^{\rm L}_{N,\epsilon}\Bigl(N{}q_{\mathds{T}}\left\langle{e^{-I}{I}}\right\rangle_{\Omega^{\rm test}}\lambda^{\rm test}y_{1}^{\rm L}\Bigr)=:M_{\rm test,1}^{\rm ideal\hspace{.05cm}L},
\end{split}\end{equation}
where we have introduced $\lambda^{\rm test}=\langle{e^{-I}{I}\sin{\theta}\cos{\phi}}\rangle_{\Omega^{\rm H}}\bigl/\langle{e^{-I}{I}}\rangle_{\Omega^{\rm H}}$ for $\Omega^{\rm H}=\bigcup_{j\in\Gamma^{\rm test}}{\Omega_{j}^{\rm H}}$.
\subsection{Upper bound on the phase-error rate}\label{error_yield}
Once again, from a reverse Kato bound of Appendix~\ref{Kato} and an argument formally identical to that of Appendix~\ref{example}, one can establish that the number $m_{\rm test,1}^{\rm ideal}$ of single-photon error counts in $\mathcal{X}^{\rm test}$ triggered by perfectly prepared states satisfies
\begin{equation}\begin{split}
&m_{\rm test,1}^{\rm ideal}\overset{\epsilon[2(d_{\rm key}+d_{\rm test})+1]}<\\
&\bar{K}^{\rm U}_{N,\epsilon}\Bigl(N{}q_{\mathds{T}}\left\langle{e^{-I}{I}}\right\rangle_{\Omega^{\rm test}}\lambda^{\rm test}e_{1}^{\mathrm{ideal}\hspace{.05cm}\mathrm{U}}\Bigr)=:m_{\rm test,1}^{\rm ideal\hspace{.05cm}U},
\end{split}\end{equation}
where the error term $\epsilon[2(d_{\rm key}+d_{\rm test})+1]$ arises from the usual composition of errors.

Let us now derive an upper bound on the single-photon PHER $e^{\rm (ph)}_{1}=m^{(\rm ph)}_{\rm key,1}/M_{\rm key,1}$, where $m^{(\rm ph)}_{\rm key,1}$ is the RV describing the number of single-photon phase errors in $\mathcal{X}^{\rm key}$. As found in Sec.~\ref{post-selection}, these errors would arise from uniformly preparing $\ket{\rm H}$ and $\ket{\rm V}$ and measuring them in the test basis, which means that the ratio $m_{\rm test,1}^{\rm ideal}/M_{\rm test,1}^{\rm ideal}$ is an unbiased estimator of $e^{\rm (ph)}_{1}$ in virtue of the standard random sampling argument invoked in the ideal BB84 protocol. To be precise, it follows from Serfling's inequality~\cite{Serfling} that
\begin{equation}\label{Serfling_errors}
m^{(\rm ph)}_{\rm key,1}\overset{\epsilon}{<}m_{\rm test,1}^{\rm ideal}\frac{M_{\rm key,1}}{M_{\rm test,1}^{\rm ideal}}+\Upsilon\left(M_{\rm key,1},M_{\rm test,1}^{\rm ideal},\epsilon\right),
\end{equation}
for $\Upsilon(x,y,z)=\sqrt{(x+y)x(y+1)\ln\left(z^{-1}\right)/2y^2}$. Now, making monotonicity considerations separately for each term in the r.h.s. of Eq.~(\ref{Serfling_errors}), it follows that
\begin{equation}\begin{split}
&m^{(\rm ph)}_{\rm key,1}\overset{\epsilon[2(d_{\rm key}+d_{\rm test})+4]}{<}\\
&m_{\rm test,1}^{\mathrm{ideal}\hspace{.05cm}\mathrm{U}}\frac{M_{\rm key,1}^{\rm U}}{M_{\rm test,1}^{\rm ideal\hspace{.05cm}L}}+\Upsilon\left(M_{\rm key,1}^{\rm U},M_{\rm test,1}^{\rm ideal\hspace{.05cm}L},\epsilon\right)=:m^{(\rm ph)\hspace{.05cm}U}_{\rm key,1},
\end{split}\end{equation}
where the usual composition of errors is applied as well. Consequently, the desired bound on the PHER is
\begin{equation}\label{final_bound}
e^{\rm (ph)}_{1}\overset{\epsilon[2(d_{\rm key}+d_{\rm test})+5]}{<}\frac{m^{(\rm ph)\hspace{.05cm}U}_{\rm key,1}}{M_{\rm key,1}^{\rm L}}=:e^{\rm (ph)\hspace{.05cm}\rm U}_{1},
\end{equation}
$M_{\rm key,1}^{\rm L}$ being given in Eq.~(\ref{Kato_reverse_2}). To sum up, the error bound in Eq.~(\ref{final_bound}), which matches the overall PE error, $\epsilon_{\rm PE}=\epsilon[2(d_{\rm key}+d_{\rm test})+5]$, includes $2(d_{\rm key}+d_{\rm test})$ direct Kato bounds for the gains $\{Q^{\rm key}_{j}\}_{j\in\Gamma^{\rm key}}$ and the error gains $\{E^{\rm test}_{j}\}_{j\in\Gamma^{\rm test}}$, four reverse Kato bounds for $M_{\rm key,1}^{\rm L}$, $M_{\rm key,1}^{\rm U}$, $M_{\rm test,1}^{\rm ideal\hspace{.05cm}L}$ and $m_{\rm test,1}^{\rm ideal\hspace{.05cm}U}$, and one Serfling bound for $e^{\rm (ph)\hspace{.05cm}\rm U}_{1}$.
\section{Performance}\label{performance}
In this section, we illustrate the rate-distance performance of the proposed passive protocol. The secret key rate is defined as $K=\max\{l,0\}/N$, where we recall that the key length $l$ is calculated according to Eq.~(\ref{keylength}).

In the absence of experimental data, we consider a standard channel and detector model given in Appendix~\ref{channel_model}, specified by the overall detection efficiency of the system, $\eta$, and the dark count probability of Bob's detectors, $p_{\rm d}$. The $\eta$ parameter factors as $\eta=\eta_{\rm Bob}\eta_{\rm ch}$, where $\eta_{\rm Bob}$ stands for the detection efficiency of Bob's detectors and $\eta_{\rm ch}$ stands for the channel transmittance, modeled as $\eta_{\rm ch}=10^{-\alpha_{\rm att}{}L/10}$ in terms of the attenuation coefficient of the channel, $\alpha_{\rm att}$, and the transmission length, $L$. For comparison purposes, we use the same values as in~\cite{Zapatero}: $\eta_{\rm det}=65\%$, $\alpha=0.2$ dB/km and $p_{\rm d}=10^{-6}$.


The threshold photon number of the LPs is set to $n_{\rm cut}=4$, resulting in a negligible loss in performance according to our simulations. For illustration purposes, the numbers of decoy settings are set to $d_{\rm key}=d_{\rm test}=4$, in such a way that we can denote $\Gamma^{\rm key}=\{\alpha,\beta,\gamma,\delta\}$ and $\Gamma^{\rm test}=\{\alpha',\beta',\gamma',\delta'\}$. We observe that using four settings instead of three provides a better robustness to loss, and we remark that, contrarily to what happens with an active setup, increasing the number of decoys does not require to modify the hardware in the passive scenario. To be precise, for the key basis, we select consecutive intervals of the form $I_{\delta}/4\nu{}t=[0,w)$, $I_{\gamma}/4\nu{}t=[w,2w)$, $I_{\beta}/4\nu{}t=[2w,3w)$ and $I_{\alpha}/4\nu{}t=[3w,1)$, $w$ denoting a fixed width parameter. For the test basis, we use overlapping intervals instead, \textit{i.e.} $I_{\delta'}/4\nu{}t=[0,w)$, $I_{\gamma'}/4\nu{}t=[0,2w)$, $I_{\beta'}/4\nu{}t=[0,3w)$ and $I_{\alpha'}/4\nu{}t=[0,1)$, maintaining the width parameter $w$ for simplicity. Evidence based on multiple numerical trials indicates that using overlapping intervals in the test basis seems beneficial, but apparently it is not so for the key basis. On top of it, we perform a mild Monte Carlo optimization in the width parameter, $w$, the test-basis measurement probability, $q_{\mathds{T}}$, the output intensity value $\nu{}t$, and the angular widths of the post-selection regions, $\Delta\theta'$, $\Delta\theta$ and $\Delta\phi$. In this regard, we recall that a minor advantage may be achieved by using different ---or even unconstrained--- decoy widths in each basis, or further increasing the number of decoys. Notwithstanding, a preliminary analysis seems to indicate that there is not much room for improvement in this respect. 

As for the finite-key parameters (introduced in Appendix~\ref{security_claims}), we set the fixed error tolerance of the PE bounds to $\epsilon=10^{-20}$ and take $\epsilon_{\rm PA}=\delta=\epsilon$, leading to an overall PE error of $\epsilon_{\rm PE}=21\epsilon$ and a secrecy parameter of $\epsilon_{\rm sec}=\sqrt{\epsilon_{\rm PE}}+\epsilon_{\rm PA}+\delta\approx{}4.58\times{}10^{-10}$. The correctness error probability is set to $\epsilon_{\rm cor}=\epsilon=10^{-20}$ as well. Lastly, all the necessary guesses for Kato’s inequality are selected under the assumption of a perfect characterization of the
experimental setup. Note, however, that this assumption can be easily removed in practice with a minuscule penalty on the key rate, by carefully modelling the channel or
by using the outcomes of previous experiments.

To finish with, in order to account for the EC leakage (which would be quantified precisely in a real experiment), we use a standard model described in Appendix~\ref{channel_model}. The model relies on an EC efficiency parameter, $f_{\rm EC}$, which we set to a typical value of $f_{\rm EC}=1.16$ in the simulations.

\begin{figure}[!htbp]
	\centering 
	\includegraphics[width=8.5cm,height=6cm]{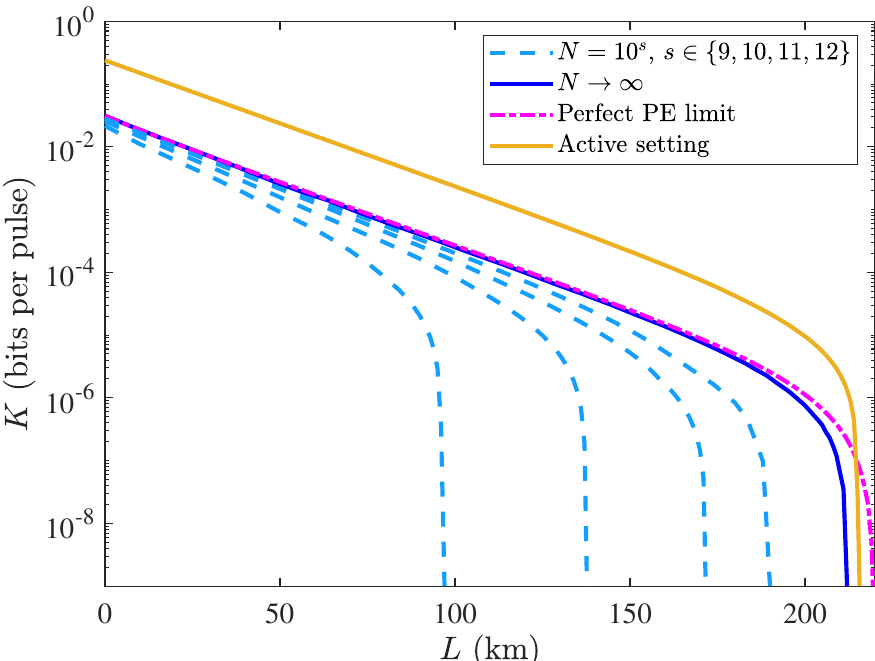}\\
	\caption{Rate-distance performance of the fully passive decoy-state BB84 protocol. The finite-key settings, the experimental parameters and the optimization method are described in the main text. Dashed light-blue lines: finite regime ($N$ = number of transmission rounds); solid dark-blue line: asymptotic regime; dashed-dotted pink line: perfect PE limit; solid yellow line: asymptotic performance of an active protocol.}
	\label{fig:performance}
\end{figure}

The rate-distance plots are shown in Fig.~\ref{fig:performance}, contemplating different numbers of transmission rounds, $N=10^{s}$ with $s\in\{9,10,11,12\}$ (dashed light-blue lines), and the asymptotic limit $N\to{}\infty$ (solid dark-blue line) as well. Since $N\to{}\infty$ favours arbitrarily small test-basis regions with the considered protocol, in that case we conservatively assume finite minimum sizes $\Delta\theta=\Delta\phi=0.1$ and $w=5\times{}10^{-3}$, and optimize the key-basis regions via $\nu{}t$ and $\Delta\theta'$ (also, the consideration of overlapping intervals becomes essentially irrelevant in this regime). By doing this, we observe that the asymptotic key rate $K_{\infty}$ is fairly close to the perfect PE limit (dashed-dotted pink line), in which the estimates of the LPs are replaced by the actual values yielded by the channel model, and the key-basis regions are optimized as well. Importantly, the closeness to the perfect PE limit is not a consequence of the noise-suppressing constraints, which only have an (indeed modest) impact in the finite-size secret key rates ---\textit{e.g.} about a 4\% increase for $N=10^{12}$ and a 16\% for $N=10^{9}$ in our simulations---. Intuitively speaking, the bound on the PHER estimated without noise suppression at all is already ``low enough", in the sense that further decreasing it (even by orders of magnitude) does not have a great impact on the secret key rate. Because of this, it is the size of the key regions ---which governs the sifting and the BER of the protocol--- that essentially determines the secret key rate. 

For the sake of comparison, we have also plotted $K_{\infty}$ for a standard active decoy-state BB84 protocol with three intensity settings (solid yellow line), assuming the same channel model and experimental inputs. The figure shows that, with the considered protocol design and PE method, the key rate of the active approach only exceeds that of the passive approach by less than one order of magnitude. As already pointed out in~\cite{Mike,Zapatero}, this drop is mainly due to the the extra sifting and the higher BER inherent to the passive scenario.
\section{Conclusions and outlook}\label{Outlook}
Passive QKD replaces all active modulation in the QKD hardware by a fixed quantum mechanism and post-selection, in so benefitting from a considerable simplicity, immunity to modulator side channels, and conceivably higher clock rates. In this work, we have derived finite-key security bounds for a fully passive decoy-state BB84 protocol. For this purpose, we have sharpened the security analysis of~\cite{Zapatero} and generalized it to the finite key regime, considering the post-selection strategy originally proposed in~\cite{Mike}. As a result, the secret key rate per pulse achievable with our passive scheme is the closest ever reported to the corresponding key rate with an active setup. On top of it, the bounds derived here have been applied to a pioneering fully passive QKD demonstration~\cite{Lu}, showing that passive QKD solutions are ready for their deployment.

On another note, different avenues could be explored to improve this work. As an example, for moderate numbers of transmission rounds, a slight advantage might be achieved replacing the trace-distance constraints with tighter nonlinear constraints provided by the Cauchy-Schwarz inequality, and running semidefinite programs ---rather than linear programs--- for the parameter estimation (see \textit{e.g.}~\cite{correlations,Xoel}). Interestingly as well, one could attempt to meet existing experimental limitations by relaxing the assumptions on the passive source, say, regarding the perfect
visibility interference or the noiseless measurements in Alice’s module.

\section{Acknowledgements}\label{Acknowledgements}
VZ acknowledges fruitful discussions with Guillermo Currás and Álvaro Navarrete. Both VZ and MC acknowledge support from Cisco Systems Inc., the Galician Regional Government (consolidation of Research Units: AtlantTIC), the Spanish Ministry of Economy and Competitiveness (MINECO), the Fondo Europeo de Desarrollo Regional (FEDER) through the grant No. PID2020-118178RB-C21, MICIN with funding from the European Union NextGenerationEU (PRTR-C17.I1) and the Galician Regional Government with own funding through the “Planes Complementarios de I+D+I con las Comunidades Autónomas” in Quantum Communication, the European Union’s Horizon Europe Framework Programme under the Marie Sklodowska-Curie Grant No. 101072637 (Project QSI) and the project “Quantum Security Networks Partnership” (QSNP, grant agreement No. 101114043).

\appendix
\section{Trace distance argument}\label{TD_argument}
Let $\rho$ and $\tau$ be two density matrices of a quantum system of dimension $d$. The TD argument states that
\begin{equation}
D(\rho,\tau)=\max_{\hat{O}}\left\{\mathrm{Tr}\left[\hat{O}(\rho-\tau)\right]\right\},
\end{equation}
where the maximization is taken over all positive operators $\hat{O}\leq{I}$~\cite{Nielsen}. Notably, from the definition of the TD it follows that $D(\rho,\tau)=\sum_{i=1}^{d}\abs{\lambda_{i}}$, where the $\lambda_{i}$ are the eigenvalues of $\rho-\tau$.
\section{Matrix representation of the Fock states}\label{representation}
Here we provide a matrix representation of the Fock states prepared by the passive source, reproducing the approach in~\cite{Zapatero} (the reader is referred to that work for more details). Precisely, let
\begin{equation}
\mathcal{B}_{n}=\left\{\ket{n-k,k}=\frac{a^{\dagger n-k}_{\rm R}a^{\dagger k}_{\rm L}}{\sqrt{(n-k)!k!}}\ket{\rm vac},\ k=0\ldots{}n\right\},
\end{equation}
where $a^{\dagger}_{\rm R}$ ($a^{\dagger}_{\rm L}$) is the creation operator of the right-handed (left-handed) circular polarization mode. $\mathcal{B}_{n}$ is an orthonormal basis of the Hilbert space of of $n$ indistinguishable photons distributed across two modes. Given $\mathcal{B}_{n}$, one can make use of the canonical isomorphism
\begin{equation}\begin{split}
&\ket{n,0}\rightarrow{}[1\hspace{.05cm}0\ldots{}\hspace{.05cm}0]^{t},\hspace{.1cm}\ket{n-1,1}\rightarrow{}[0\hspace{.05cm}1\ldots{}\hspace{.05cm}0]^{t}\hspace{.1cm}\ldots{}\\
&\ket{1,n-1}\rightarrow{}[0\hspace{.05cm}\ldots{}\hspace{.05cm}1\hspace{.05cm}0]^{t},\hspace{.1cm}\ket{0,n}\rightarrow{}[0\hspace{.05cm}\ldots{}\hspace{.05cm}0\hspace{.05cm}1]^{t},
\end{split}\end{equation}
to obtain the density matrices of the Fock states $\sigma^{\rm key}_{j,n}=\langle{e^{-\mu}{\mu}^{n}/n!\ketbra{n}{n}_{\theta,\phi}}\rangle_{\Omega^{\rm key}_{j}}/\langle{e^{-\mu}{\mu}^{n}/n!}\rangle_{\Omega^{\rm key}_{j}}$ and $\sigma^{\rm H}_{j,n}=\langle{e^{-\mu}{\mu}^{n}/n!\ketbra{n}{n}_{\theta,\phi}}\rangle_{\Omega^{\rm H}_{j}}/\langle{e^{-\mu}{\mu}^{n}/n!}\rangle_{\Omega^{\rm H}_{j}}$ (\textit{i.e.}, the ones required for the LPs in the main text) by recalling that
\begin{equation}\label{fock_theta_phi}
\ket{n}_{\theta,\phi}=\frac{1}{\sqrt{n!}}\left[\cos\left(\frac{\theta}{2}\right)a^{\dagger}_{\rm R}+e^{i\phi}\sin\left(\frac{\theta}{2}\right)a^{\dagger}_{\rm L}\right]^{n}\ket{\rm vac}.
\end{equation}
For instance, the $(r,s)$-th entry of $\sigma^{\rm key}_{j,n}$ is computed as $\bra{n-r+1,r-1}\sigma^{\rm key}_{j,n}\ket{n-s+1,s-1}$ for $r,s=1\ldots{}n+1$, which in fact vanishes for the off-diagonal terms ($r\neq{}s$).
\section{Security claims}\label{security_claims}
\begin{definition}[Correctness]
A pair of keys $(\textbf{S}_{\rm A},\textbf{S}_{\rm B})$ is $\epsilon_{\rm cor}$-correct if $\Pr\left[\textbf{S}_{\rm A}\neq{}\textbf{S}_{\rm B}\right]\leq{}\epsilon_{\rm cor}$.
\end{definition}
\begin{proposition}[Correctness claim]
The output keys of the protocol are $\epsilon_{\rm cor}$-correct.
\end{proposition}
\begin{proof}
Let $\bot$ ($\top$) denote the abortion (non-abortion) event. Assuming that the protocol outputs two identical default symbols in case of abortion, \textit{i.e.} $\Pr\left[\textbf{S}_{\rm A}\neq{}\textbf{S}_{\rm B}\ \wedge\ \bot\right]=0$, and hence correctness follows if $\Pr\left[\textbf{S}_{\rm A}\neq{}\textbf{S}_{\rm B}\ \wedge\ \top\right]\leq{}\epsilon_{\rm cor}$. Let $h_{\rm EV}$ ($h_{\rm PA}$) denote the two-universal hash function deployed for EV (PA), and let $\textbf{K}_{\rm A}$ ($\textbf{K}_{\rm B}$) stand for Alice's sifted key (Bob's error-corrected key), such that $\textbf{S}_{\rm A}=h_{\rm PA}(\textbf{K}_{\rm A})$ ($\textbf{S}_{\rm B}=h_{\rm PA}(\textbf{K}_{\rm B})$). We have
\begin{equation}\begin{split}
&\Pr\left[\textbf{S}_{\rm A}\neq{}\textbf{S}_{\rm B}\ \wedge\ \top\right]\leq{}\\
&\Pr\left[h_{\rm PA}(\textbf{K}_{\rm A})\neq{}h_{\rm PA}(\textbf{K}_{\rm B})\ \wedge\ {}h_{\rm EV}(\textbf{K}_{\rm A})=h_{\rm EV}(\textbf{K}_{\rm B})\right]\leq{}\\
&\Pr\left[\textbf{K}_{\rm A}\neq\textbf{K}_{\rm B}\ \wedge\ h_{\rm EV}(\textbf{K}_{\rm A})=h_{\rm EV}(\textbf{K}_{\rm B})\right]=\\
&\Pr\left[\textbf{K}_{\rm A}\neq\textbf{K}_{\rm B}\right]\Pr\left[h_{\rm EV}(\textbf{K}_{\rm A})=h_{\rm EV}(\textbf{K}_{\rm B})\ |\ \textbf{K}_{\rm A}\neq\textbf{K}_{\rm B}\right]\leq{}\\
&\Pr\left[h_{\rm EV}(\textbf{K}_{\rm A})=h_{\rm EV}(\textbf{K}_{\rm B})\ |\ \textbf{K}_{\rm A}\neq\textbf{K}_{\rm B}\right]\leq{}\epsilon_{\rm cor}.
\end{split}\end{equation}
The first inequality follows because $\{\top\}\implies{}\{h_{\rm EV}(\textbf{K}_{\rm A})=h_{\rm EV}(\textbf{K}_{\rm B})\}$, the second inequality follows because $\{h_{\rm PA}(\textbf{K}_{\rm A})\neq{}h_{\rm PA}(\textbf{K}_{\rm B})\}\implies{}\{\textbf{K}_{\rm A}\neq\textbf{K}_{\rm B}\}$, the third inequality follows because $\Pr\left[\textbf{K}_{\rm A}\neq\textbf{K}_{\rm B}\right]\leq{}1$ and the fourth inequality follows by the definition of two-universal hashing~\cite{Wegman}.
\end{proof}
\begin{definition}[Secrecy]
 Let $\rho_{\rm SE}=\sum_{s_{\rm A}}\Pr[\textbf{S}_{\rm A}=s_{\rm A}]\ketbra{s_{\rm A}}{s_{\rm A}}_{\rm A}\otimes{}\rho^{s_{\rm A}}_{\rm E}$ be a bipartite cq state describing a classical register $A$ ---modelling a RV $\textbf{S}_{\rm A}$--- and a quantum system $E$. $\textbf{S}_{\rm A}$ is $\epsilon_{\rm sec}$-secret from $E$ if $\min_{\displaystyle{\sigma_{\rm E}}}D\left(\rho_{\rm SE},\omega_{\rm S}\otimes{}\sigma_{\rm E}\right)\leq{}\epsilon_{\rm sec}$, where $\omega_{\rm S}$ denotes the fully mixed state on A.
 \end{definition}
 \begin{proposition}[Secrecy claim]
Alice's output key in the protocol is $\left(\sqrt{\epsilon_{\rm PE}}+\epsilon_{\rm PA}+\delta\right)$-secret from $E$ for
\begin{equation}\label{keylength}
l=\left\lfloor{M_{\rm key,1}^{\rm L}\left[1-h\left(e^{(\rm ph)\hspace{.05cm}\rm U}_{1}\right)\right]-\lambda_{\rm EC}-\log\left(\frac{1}{2\epsilon_{\rm cor}\epsilon_{\rm PA}^{2}\delta}\right)}\right\rfloor,
\end{equation}
where $\delta>0$ is arbitrary.
\end{proposition}
\begin{proof}
The secrecy claim is established here on the basis of the quantum leftover hash lemma~\cite{LHL}. Below we recast the statement of the lemma provided in~\cite{Tomamichel}.
\begin{theorem}[Quantum leftover hash lemma]
Let $\sigma_{\rm KE}=\sum_{k_{\rm A}}\Pr[\textbf{K}_{\rm A}=k_{\rm A}]\ketbra{k_{\rm A}}{k_{\rm A}}_{\rm A}\otimes{}\sigma^{k_{\rm A}}_{\rm E}$ be a bipartite cq state, and let $C$ be a classical RV publicly revealed. Applying random two-universal hashing on the register $A$, an $(\epsilon+\epsilon_{\rm PA})$-secret-from-$E$ key of length
\begin{equation}\label{PA}
l=\left\lfloor{H_{\rm min}^{\epsilon}(\textbf{K}_{\rm A}|CE)_{\sigma}-\log\left(\frac{1}{4\epsilon_{\rm PA}^{2}}\right)}\right\rfloor
\end{equation}
can be extracted, where $H_{\rm min}^{\epsilon}(\textbf{K}_{\rm A}|CE)_{\sigma}$ stands for the $\epsilon$-smooth conditional min-entropy of $\textbf{K}_{\rm A}$ given access to $E$ and $C$, evaluated in $\sigma_{\rm KE}$.
\end{theorem}
For our purposes, $\sigma_{\rm KE}$ is a fully general description of Alice's and Eve's cq state prior to PA, from which the final state $\rho_{\rm SE}$ is obtained in the protocol via random two-universal hashing on the register $A$. According to Theorem 1, the goal is to derive a lower bound on $H_{\rm min}^{\epsilon}(\textbf{K}_{\rm A}|CE)_{\sigma}$, which is the purpose of the rest of the proof.\\

The RV $C$ matches the information revealed by EC and EV in the protocol. Since EC discloses $\lambda_{\rm EC}$ syndrome bits at most and EV discloses $\lceil{\log(1/\epsilon_{\rm cor})}\rceil<\log\left({2}/{\epsilon_{\rm cor}}\right)$ EV-tag bits at most, it follows that
\begin{equation}\label{chain_rule}
H_{\rm min}^{\epsilon}(\textbf{K}_{\rm A}|CE)_{\sigma}\geq{}H_{\rm min}^{\epsilon}(\textbf{K}_{\rm A}|E)_{\sigma}-\lambda_{\rm EC}-\log\left(\frac{2}{\epsilon_{\rm cor}}\right),
\end{equation}
by applying a chain rule for smooth min-entropies~\cite{Tomamichel}. Coming next, we decompose the RV $\textbf{K}_{\rm A}$ as $\textbf{K}_{\rm A}=\textbf{K}_{\rm A}^{1}\textbf{K}_{\rm A}^{\rm rest}$, where $\textbf{K}_{\rm A}^{1}$ describes those sifted key bits corresponding to single-photon detection events, and $\textbf{K}_{\rm A}^{\rm rest}$ describes the rest of the bits. From this decomposition and a generalized chain rule given in~\cite{Vitanov}, it follows that
\begin{equation}\label{chain_rule_2}
H_{\rm min}^{\epsilon}(\textbf{K}_{\rm A}|E)_{\sigma}\geq{}H_{\rm min}^{\epsilon-\delta}(\textbf{K}_{\rm A}^{1}|E)_{\sigma}-\log\left(\frac{1}{\delta}\right)
\end{equation}
for all $\delta\in(0,\epsilon)$, and plugging Eq.~(\ref{chain_rule_2}) into Eq.~(\ref{chain_rule}) yields
\begin{equation}\label{semifinal}
H_{\rm min}^{\epsilon}(\textbf{K}_{\rm A}|CE)_{\sigma}\geq{}H_{\rm min}^{\epsilon-\delta}(\textbf{K}_{\rm A}^{1}|E)_{\sigma}-\lambda_{\rm EC}-\log\left(\frac{2}{\epsilon_{\rm cor}\delta}\right).
\end{equation}
Thus, all that remains is to lower-bound the term $H_{\rm min}^{\epsilon-\delta}(\textbf{K}_{\rm A}^{1}|E)_{\sigma}$. For this purpose, we make use of the entropic uncertainty relation~\cite{EUR}, which we present below in the restricted case of projective measurements (that suffice for our analysis).
\begin{theorem}[Entropic uncertainty relation]
Let $\tau_{\rm ABE}$ be a tripartite quantum state, and let $\mathds{K}=\{\ketbra{k}{k}_{\rm A}\}$ and $\mathds{T}=\{\ketbra{t}{t}_{\rm A}\}$ be two projective measurements on A, $\textbf{K}$ and $\textbf{T}$ denoting the RVs associated to their measurement outcomes. Then, for all $\varepsilon\geq{}0$,
\begin{equation}
H_{\rm min}^{\varepsilon}(\textbf{K}|E)_{\tau}+H_{\rm max}^{\varepsilon}(\textbf{T}|B)_{\tau}\geq{}q,
\end{equation}
where $H_{\rm min}^{\varepsilon}(\textbf{K}|E)_{\tau}$ stands for the $\epsilon$-smooth conditional min-entropy of $\textbf{K}$ given $E$ evaluated in the bipartite cq state $\tau_{\rm KE}=\sum_{k}\Pr[\textbf{K}=k]\ketbra{k}{k}_{\rm A}\otimes{}\tau^{k}_{\rm E}$, $H_{\rm max}^{\varepsilon}(\textbf{T}|B)_{\tau}$ stands for the $\epsilon$-smooth conditional max-entropy of $\textbf{T}$ given $B$ evaluated in the bipartite cq state $\tau_{\rm TB}=\sum_{t}\Pr[\textbf{T}=t]\ketbra{t}{t}_{\rm A}\otimes{}\tau^{t}_{\rm B}$, and $q=\log\frac{1}{c}$ for $c=\max_{k,t}\abs{\braket{k}{t}}^{2}$.
\end{theorem}

Let $\tau^{1}_{\rm ABE}$ be the cqq state describing:
\begin{enumerate}
\item{The sub-register $A^{1}$ of $A$ that models the RV $\textbf{K}_{\rm A}^{1}$ ($\textbf{T}_{\rm A}^{1}$) via a projective measurement $\mathds{K}$ ($\mathds{T}$) in the key basis (test basis).}
\item{Bob's quantum side information on $A^{1}$, say $B^{1}$.}
\item{Eve's quantum side information on $A^{1}$, say $E^{1}$.}
\end{enumerate}
Crucially, this choice of $\tau^{1}_{\rm ABE}$ is such that $\sigma_{\rm KE}$ provides an extension of $\tau^{1}_{\rm KE}=\sum_{k^{1}_{\rm A}}\Pr[\textbf{K}^{1}_{\rm A}=k^{1}_{\rm A}]\ketbra{k^{1}_{\rm A}}{k^{1}_{\rm A}}_{\rm A}\otimes{}\tau^{k^{1}_{\rm A}}_{\rm E}$, as the former is related to the latter by suitable partial tracing. This will be relevant to link the entropies of both states. 

Also, we have that $\mathds{K}=\left\{\ketbra{k_{\rm A}^{1}}{k_{\rm A}^{1}}\right\}$ and $\mathds{T}=\left\{\ketbra{t_{\rm A}^{1}}{t_{\rm A}^{1}}\right\}$ for $k_{\rm A}^{1}\in\{\mathrm{R,L}\}^{M_{\rm key,1}}$ and $t_{\rm A}^{1}\in\{\mathrm{H,V}\}^{M_{\rm key,1}}$, where $M_{\rm key,1}$ denotes the number of bits in the register $A^{1}$ (\textit{i.e.,} the size of $\textbf{K}_{\rm A}^{1}$). It follows then that $c=\max_{k,t}\abs{\braket{k_{\rm A}^{1}}{t_{\rm A}^{1}}}^{2}=2^{-M_{\rm key,1}}$, and in virtue of Theorem 2 we have that
\begin{equation}\label{EUR}
H_{\rm min}^{\varepsilon}(\textbf{K}_{\rm A}^{1}|E^{1})_{\tau^{1}}\geq{}M_{\rm key,1}-H_{\rm max}^{\varepsilon}(\textbf{T}_{\rm A}^{1}|B^{1})_{\tau^{1}}.
\end{equation}
Eq.~(\ref{EUR}) sets a fundamental trade-off between Eve's and Bob's predictive powers on the outcomes of mutually unbiased measurements $\mathds{K}$ and $\mathds{T}$ performed on $A^{1}$. Importantly, this trade-off is established \emph{a priori} of any protocol execution.

Similarly, the PE procedure explained in Section~\ref{secret_key_parameters} also provides an \emph{a priori} statistical constraint on the RVs $M_{\rm key,1}$ and $m^{(\rm ph)}_{\rm key,1}$ on $\tau^{1}_{\rm ABE}$. Namely, that the event
\begin{equation}
\omega=\left\{M_{\rm key,1}\notin\left(M_{\rm key,1}^{\rm L},M_{\rm key,1}^{\rm U}\right)\ \cup\ m^{(\rm ph)}_{\rm key,1}\geq{}m^{(\rm ph)\hspace{.05cm}U}_{\rm key,1}\right\}
\end{equation}
(where the bounding RVs $M_{\rm key,1}^{\rm L}$, $M_{\rm key,1}^{\rm U}$ and $m^{(\rm ph)\hspace{.05cm}U}_{\rm key,1}$ are defined in Sec.~\ref{secret_key_parameters}), satisfies $\Pr[\omega]\leq{}\epsilon_{\rm PE}$. This result allows to lower bound $H_{\rm min}^{\varepsilon}(\textbf{K}_{\rm A}^{1}|E^{1})_{\tau^{1}}$ for a suitably chosen $\varepsilon$ using the following lemma~\cite{Leverrier}.
\begin{lemma}
Let $\rho$ be a quantum state and let $\omega$ be an arbitrary event. If $\rho$ is such that $\Pr[\omega]\leq{}\epsilon$, there exists another state $\tilde\rho$, $\sqrt{\epsilon}$-close to $\rho$ in purified distance, \textit{i.e.} verifying $P(\rho,\tilde\rho)\leq{}\sqrt{\epsilon}$, for which $\Pr[\omega]=0$.
\end{lemma}
In virtue of Lemma 1 (the reader is referred to~\cite{Leverrier} for a definition of the purified distance), there exists a quantum state $\tilde{\tau}^{1}_{\rm ABE}$ with $P({\tau}^{1}_{\rm ABE},\tilde{\tau}^{1}_{\rm ABE})\leq{}\sqrt{\epsilon_{\rm PE}}$ for which $\Pr[\omega]=0$. In what follows we use this state as an artifact to lower bound $H_{\rm min}^{\sqrt{\epsilon_{\rm PE}}}(\textbf{K}_{\rm A}^{1}|E^{1})_{\tau^{1}}$.

First of all, particularizing Eq.~(\ref{EUR}) for $\tilde{\tau}^{1}_{\rm ABE}$ in the limit case $\varepsilon=0$, we have that
\begin{equation}\label{EUR_2}
H_{\rm min}(\textbf{K}_{\rm A}^{1}|E^{1})_{\tilde\tau^{1}}\geq{}M_{\rm key,1}-H_{\rm max}(\textbf{T}_{\rm A}^{1}|B^{1})_{\tilde\tau^{1}}.
\end{equation}
Coming next, we shall derive an upper bound on the non-smooth max-entropy $H_{\rm max}(\textbf{T}_{\rm A}^{1}|B^{1})_{\tilde{\tau}^{1}}$. For this purpose, let us consider the set of quantum signals that contribute to $\textbf{K}_{\rm A}^{1}$. In particular, let $\textbf{T}_{\rm B}^{1}$ be the RV defined by the outcome of measuring this set of signals in the test basis. Clearly, $H_{\rm max}(\textbf{T}_{\rm A}^{1}|B^{1})_{\tilde{\tau}^{1}}\leq{}H_{\rm max}(\textbf{T}_{\rm A}^{1}|\textbf{T}_{\rm B}^{1})_{\tilde{\tau}^{1}}$ because $\textbf{T}_{\rm B}^{1}$ is accessible given $B^{1}$. Moreover, $H_{\rm max}(\textbf{T}_{\rm A}^{1}|\textbf{T}_{\rm B}^{1})_{\tilde{\tau}^{1}}$ is quantified by the logarithm of the number of outcomes for $\textbf{T}_{\rm A}^{1}$ with a nonzero probability to occur. Precisely~\cite{Leverrier},
\begin{equation}\begin{split}
&H_{\rm max}(\textbf{T}_{\rm A}^{1}|\textbf{T}_{\rm B}^{1})_{\tilde{\tau}}\leq{}\\
&\max_{\displaystyle{t_{\rm B}^{1}:\Pr[\textbf{T}_{\rm B}^{1}=t_{\rm B}^{1}]>0}}\log{\left\lvert\biggl\{t_{\rm A}^{1}\ :\ \Pr\Bigl[\textbf{T}_{\rm A}^{1}=t_{\rm A}^{1}|\textbf{T}_{\rm B}^{1}=t_{\rm B}^{1}\Bigr]\biggr\}\right\rvert},
\end{split}\end{equation}
and given an upper bound $e^{(\rm ph)\hspace{.05cm}\rm U}_{1}$ on the ratio of errors $e^{(\rm ph)}_{1}$ between $\textbf{T}_{\rm A}^{1}$ and $\textbf{T}_{\rm B}^{1}$ ---namely, the so-called phase-error rate--- this can be upper-bounded via~\cite{Lint}
\begin{equation}\label{max}
H_{\rm max}(\textbf{T}_{\rm A}^{1}|\textbf{T}_{\rm B}^{1})_{\tilde{\tau}}\leq{}M_{\rm key,1}\ h\left(e^{(\rm ph)\hspace{.05cm}\rm U}_{1}\right),
\end{equation}
where $h(x)$ is the binary entropy function. Plugging this into Eq.~(\ref{EUR_2}) we obtain
\begin{equation}\label{EUR_3}
H_{\rm min}(\textbf{K}_{\rm A}^{1}|E^{1})_{\tilde\tau^{1}}\geq{}M_{\rm key,1}\left[1-h\left(e^{(\rm ph)\hspace{.05cm}\rm U}_{1}\right)\right].
\end{equation}
Crucially for $\tilde\tau^{1}_{\rm ABE}$, $M_{\rm key,1}$ and $e^{(\rm ph)\hspace{.05cm}\rm U}_{1}$ can be related to observables of the protocol because of the fact that $\Pr[\omega]=0$. On the one hand, the latter implies that $M_{\rm key,1}\in\bigl(M_{\rm key,1}^{\rm L},M_{\rm key,1}^{\rm U}\bigr)$. On the other hand, it implies that $m^{(\rm ph)}_{\rm key,1}\leq{}m^{(\rm ph)\hspace{.05cm}U}_{\rm key,1}$, such that one can choose
\begin{equation}
e^{(\rm ph)\hspace{.05cm}\rm U}_{1}=\min\left\{\frac{m^{(\rm ph)\hspace{.05cm}\rm U}_{\rm key,1}}{M_{\rm key,1}^{\rm L}},\frac{1}{2}\right\}
\end{equation}
for $\tilde\tau^{1}_{\rm ABE}$. As a consequence, it follows that
\begin{equation}\label{PA_term}
H_{\rm min}(\textbf{K}_{\rm A}^{1}|E^{1})_{\tilde\tau^{1}}\geq{}M_{\rm key,1}^{\rm L}\left[1-h\left(e^{(\rm ph)\hspace{.05cm}\rm U}_{1}\right)\right].
\end{equation}
Now, since $H_{\rm min}^{\varepsilon}(\textbf{X}|Y)_{\rho}=\max_{\displaystyle{P(\rho,\tilde\rho)\leq{}\varepsilon}}H_{\rm min}(\textbf{X}|Y)_{\tilde\rho}$, it follows from Eq.~(\ref{PA_term}) that $H_{\rm min}^{\sqrt{\epsilon_{\rm PE}}}(\textbf{K}_{\rm A}^{1}|E^{1})_{\tau^{1}}\geq{}M_{\rm key,1}^{\rm L}\bigl[1-h\bigl(e^{(\rm ph)\hspace{.05cm}\rm U}_{1}\bigr)\bigr]$ by noticing that $P({\tau}^{1}_{\rm KE},\tilde{\tau}^{1}_{\rm KE})\leq{}P({\tau}^{1}_{\rm ABE},\tilde{\tau}^{1}_{\rm ABE})\leq{}\sqrt{\epsilon_{\rm PE}}$ (in the penultimate inequality we are using the fact that the purified distance is non-increasing under trace non-increasing completely positive maps~\cite{thesis}).

In fact, the exact same bound holds for $H_{\rm min}^{\sqrt{\epsilon_{\rm PE}}}(\textbf{K}_{\rm A}^{1}|E)_{\sigma}$ directly ---and not just for $H_{\rm min}^{\sqrt{\epsilon_{\rm PE}}}(\textbf{K}_{\rm A}^{1}|E^{1})_{\tau^{1}}$--- in virtue of the following lemma~\cite{QIP}.
\begin{lemma}
For any two states $\tau$ and $\tilde\tau$, and any extension $\sigma$ of $\tau$, there exists an extension $\tilde\sigma$ with $P(\sigma,\tilde\sigma)=P(\tau,\tilde\tau)$.
\end{lemma}
Particularly, since $\sigma_{\rm KE}$ is an extension of $\tau^{1}_{\rm KE}$ and $P(\tau^{1}_{\rm KE},\tilde\tau^{1}_{\rm KE})\leq{}\sqrt{\epsilon_{\rm PE}}$, there exists an extension $\tilde\sigma_{\rm KE}$ of $\tilde\tau^{1}_{\rm KE}$ such that $P(\sigma_{\rm KE},\tilde\sigma_{\rm KE})\leq{}\sqrt{\epsilon_{\rm PE}}$, meaning that $H_{\rm min}^{\sqrt{\epsilon_{\rm PE}}}(\textbf{K}_{\rm A}^{1}|E)_{\sigma}\geq{}H_{\rm min}(\textbf{K}_{\rm A}^{1}|E)_{\tilde\sigma}$. What is more, let $p_{\rm guess}(\textbf{X}|Y)_{\rho}$ denote the probability of guessing $\textbf{X}$ given access to $Y$ in an underlying quantum state $\rho$. As any other extension of $\tilde\tau^{1}_{\rm KE}$, $\tilde\sigma_{\rm KE}$ verifies $p_{\rm guess}(\textbf{K}_{\rm A}^{1}|E)_{\tilde\sigma}=p_{\rm guess}(\textbf{K}_{\rm A}^{1}|E^{1})_{\tilde\tau^{1}}$ by the ad hoc definition of $E^{1}$ ---according to which $E^{1}$ is the share of $E$ correlated to $A^{1}$ within the underlying state---. Therefore, $H_{\rm min}(\textbf{K}_{\rm A}^{1}|E)_{\tilde\sigma}=H_{\rm min}(\textbf{K}_{\rm A}^{1}|E^{1})_{\tilde\tau^{1}}$ by definition of the non-smooth min-entropy~\cite{Leverrier}.

In summary, the extension $\tilde\sigma_{\rm KE}$ provided by Lemma 2 allows to establish that
\begin{equation}\begin{split}
&H_{\rm min}^{\sqrt{\epsilon_{\rm PE}}}(\textbf{K}_{\rm A}^{1}|E)_{\sigma}\geq{}H_{\rm min}(\textbf{K}_{\rm A}^{1}|E)_{\tilde\sigma}=H_{\rm min}(\textbf{K}_{\rm A}|E^{1})_{\tilde\tau^{1}}\geq{}\\
&M_{\rm key,1}^{\rm L}\left[1-h\left(e^{(\rm ph)\hspace{.05cm}\rm U}_{1}\right)\right],
\end{split}\end{equation}
where the last inequality follows from Eq.~(\ref{PA_term}).
According to Eq.~(\ref{semifinal}), this bound on $H_{\rm min}^{\sqrt{\epsilon_{\rm PE}}}(\textbf{K}_{\rm A}^{1}|E)_{\sigma}$ implies the bound
\begin{equation}\begin{split}
&H_{\rm min}^{\sqrt{\epsilon_{\rm PE}}+\delta}(\textbf{K}_{\rm A}|CE)_{\sigma}\geq{}M_{\rm key,1}^{\rm L}\left[1-h\left(e^{(\rm ph)\hspace{.05cm}\rm U}_{\rm 1}\right)\right]\\
&-\lambda_{\rm EC}-\log\left(\frac{2}{\epsilon_{\rm cor}\delta}\right)
\end{split}\end{equation}
for all $\delta>0$. Lastly, since $\rho_{\rm SE}$ is obtained in the protocol by applying 
random two-universal hashing on the register $A$ of $\sigma_{\rm KE}$, the secrecy claim follows from Theorem 1.
\end{proof}
\begin{widetext}
\section{Kato's inequality}\label{Kato}
Let $\xi_{1},\xi_{2}\ldots{}\xi_{N}$ be a sequence of Bernoulli RVs, and let $\mathcal{F}_{0}\subseteq{}\mathcal{F}_{1}\subseteq{}\ldots{}\subseteq{}\mathcal{F}_{N}$ be an increasing chain of $\sigma$-algebras verifying $E\left(\xi_{v}|\mathcal{F}_{u}\right)=\xi_{v}$ for $v\leq{}u$. Kato's inequality~\cite{Kato} states that
\begin{equation}\label{statement}
\Pr\left[\sum_{u=1}^{N}\Pr\left(\xi_{u}=1|\mathcal{F}_{u-1}\right)-\Lambda_{N}\geq{}\left[b+a\left(\frac{2\Lambda_{N}}{N}-1\right)\right]\sqrt{N}\right]\leq{}\exp\left[\frac{-2(b^{2}-a^{2})}{\displaystyle{\left(1+\frac{4{}a}{3\sqrt{N}}\right)^{2}}}\right]
\end{equation}
for any $N\in\mathbb{N}$, $a\in\mathbb{R}$ and $b\in\mathbb{R}^{+}$ such that $b>\abs{a}$, where we have defined $\Lambda_{l}=\sum_{u=1}^{l}\xi_{u}$. Moreover, replacing $\xi_{l}$ by $1-\xi_{l}$ and $a$ by $-a$ in Eq.~(\ref{statement}), it trivially follows that
\begin{equation}\label{statement_2}
\Pr\left[\Lambda_{N}-\sum_{u=1}^{N}\Pr\left(\xi_{u}=1|\mathcal{F}_{u-1}\right)\geq{}\left[b+a\left(\frac{2\Lambda_{N}}{N}-1\right)\right]\sqrt{N}\right]\leq{}\exp\left[\frac{-2(b^{2}-a^{2})}{\displaystyle{\left(1-\frac{4{}a}{3\sqrt{N}}\right)^{2}}}\right].
\end{equation}
\subsection{Direct Kato bounds}
Eq.~(\ref{statement_2}) provides a family of lower bounds on $\sum_{u=1}^{N}\Pr\left(\xi_{u}=1|\mathcal{F}_{u-1}\right)$ given $\Lambda_{N}$, and we want to reach the tightest possible bound compatible with a certain error probability $\epsilon$. Hence, among all pairs $(a,b)$ that fulfill the $\epsilon$ error probability condition, we must pick the one that minimizes the deviation term.  However, since the latter depends on the realization $\Lambda_{N}$, this task requires to come up with a preliminary guess, say $\tilde{\Lambda}_{N}$, of $\Lambda_{N}$, with respect to which the minimization can be carried out. In short, the problem to address is
\begin{eqnarray}\label{Kato_L}
&&\argmin_{a,b}\hspace{.2cm}\left[b+a\left(\frac{2\tilde{\Lambda}_{N}}{N}-1\right)\right]\sqrt{N}\nonumber\\
&&\textup{s.t.}\hspace{.3cm}\exp\left[\frac{-2(b^{2}-a^{2})}{\left(1-\displaystyle{\frac{4{}a}{3\sqrt{N}}}\right)^{2}}\right]=\epsilon,\hspace{.2cm}b\geq{}\abs{a}.
\end{eqnarray}
Its solution is~\cite{Guille}
\begin{eqnarray}
&&a={{3\left\{9\sqrt{2}N\left(N-2\tilde{\Lambda}_{N}\right)\sqrt{-\ln{\epsilon}\left[9\tilde{\Lambda}_{N}\left(N-\tilde{\Lambda}_{N}\right)-2N\ln{\epsilon}\right]}+16N^{3/2}\ln^{2}{\epsilon}-72\tilde{\Lambda}_{N}\sqrt{N}\left(N-\tilde{\Lambda}_{N}\right)\ln {\epsilon}\right\}}\over{4\left(9N-8\ln {\epsilon}\right)\left[9\tilde{\Lambda}_{N}\left(N-\tilde{\Lambda}_{N}\right)-2N\ln {\epsilon}\right]}},\nonumber \\
&&b={{\sqrt{18Na^2-\left(16a^2-24\sqrt{N}a+9N\right)\ln{\epsilon}}}\over{3\sqrt{2N}}},
\end{eqnarray}
leading to the bound
\begin{equation}\label{claim}
\sum_{u=1}^{N}\Pr\left(\xi_{u}=1|\mathcal{F}_{u-1}\right)\overset{\epsilon}{>}K^{\rm L}_{N,\epsilon}\left(\Lambda_{N}\right)
\end{equation}
for $K^{\rm L}_{N,\epsilon}\left(\Lambda_{N}\right)=\Lambda_{N}-\left[b+a\left(2{\Lambda}_{N}/N-1\right)\right]\sqrt{N}$. Of course, the closer $\Lambda_{N}$ happens to be from $\tilde{\Lambda}_{N}$, the tighter the bound becomes, but it holds true in any case.\\

Analogously, one can tune $a$ and $b$ in Eq.~(\ref{statement}) to reach the tightest possible upper bound on $\sum_{u=1}^{N}\Pr\left(\xi_{u}=1|\mathcal{F}_{u-1}\right)$ compatible with a fixed error probability $\epsilon$, should the guess $\tilde{\Lambda}_{N}$ become true. The corresponding problem is
\begin{eqnarray}\label{Kato_U}
&&\argmin_{a,b}\hspace{.2cm}\left[b+a\left(\frac{2\tilde{\Lambda}_{N}}{N}-1\right)\right]\sqrt{N}\nonumber\\
&&\textup{s.t.}\hspace{.3cm}\exp\left[\frac{-2(b^{2}-a^{2})}{\left(1+\displaystyle{\frac{4{}a}{3\sqrt{N}}}\right)^{2}}\right]=\epsilon,\hspace{.2cm}b\geq{}\abs{a}.
\end{eqnarray}
Its solution is~\cite{Guille}
\begin{eqnarray}
&&a={{3\left\{9\sqrt{2}N\left(N-2\tilde{\Lambda}_{N}\right)\sqrt{-\ln{\epsilon}\left[9\tilde{\Lambda}_{N}\left(N-\tilde{\Lambda}_{N}\right)-2N\ln{\epsilon}\right]}-16N^{3/2}\ln^{2}{\epsilon}+72\tilde{\Lambda}_{N}\sqrt{N}\left(N-\tilde{\Lambda}_{N}\right)\ln {\epsilon}\right\}}\over{4\left(9N-8\ln {\epsilon}\right)\left[9\tilde{\Lambda}_{N}\left(N-\tilde{\Lambda}_{N}\right)-2N\ln{\epsilon}\right]}},\nonumber \\
&&b={{\sqrt{18Na^2-\left(16a^2+24\sqrt{N}a+9N\right)\ln{\epsilon}}}\over{3\sqrt{2N}}},
\end{eqnarray}
leading to the bound
\begin{equation}
\sum_{u=1}^{N}\Pr\left(\xi_{u}=1|\mathcal{F}_{u-1}\right)\overset{\epsilon}{<}K^{\rm U}_{N,\epsilon}\left(\Lambda_{N}\right)
\end{equation}
for $K^{\rm U}_{N,\epsilon}\left(\Lambda_{N}\right)=\Lambda_{N}+\left[b+a\left(2{\Lambda}_{N}/N-1\right)\right]\sqrt{N}$.
\subsection{Reverse Kato bounds}
Similar arguments allow to bound $\Lambda_{N}$ given $S_{N}=\sum_{u=1}^{N}\Pr\left(\xi_{u}=1|\mathcal{F}_{u-1}\right)$, based on a preliminary guess $\tilde{S}_{N}$ of the latter to pick suitable values of $a$ and $b$.

On the one hand, from Eq.~(\ref{statement}) one can readily show that $\Lambda_{N}\overset{\epsilon}{>}\bar{K}^{\rm L}_{N,\epsilon}\left(S_{N}\right)$ for $\bar{K}^{\rm L}_{N,\epsilon}\left(S_{N}\right)=\left[\sqrt{N}S_{N}+N(a-b)\right]/(2a+\sqrt{N})$, where $(a,b)$ is the solution to
\begin{eqnarray}\label{Kato_bar_L}
&&\argmax_{a,b}\hspace{.2cm}\frac{\sqrt{N}\tilde{S}_{N}+N(a-b)}{2a+\sqrt{N}}\nonumber\\
&&\textup{s.t.}\hspace{.3cm}\exp\left[\frac{-2(b^{2}-a^{2})}{\left(1+\displaystyle{\frac{4{}a}{3\sqrt{N}}}\right)^{2}}\right]=\epsilon,\hspace{.2cm}b\geq{}\abs{a},
\end{eqnarray}
given by~\cite{Navs}
\begin{eqnarray}\label{ab_L}
&&a=\frac{3\sqrt{N}\left\{9\left(N-2\tilde{S}_{N}\right)\sqrt{N\ln{\epsilon}\left[N\ln{\epsilon}-18\tilde{S}_{N}\left(N-\tilde{S}_{N}\right)\right]}-4N\ln^{2}{\epsilon}-9\left(8\tilde{S}_{N}^2-8N\tilde{S}_{N}+3N^2\right)\ln {\epsilon}\right\}}{4\left\{4N\ln^{2}{\epsilon}+36\left(2\tilde{S}_{N}^2-2N\tilde{S}_{N}+N^2\right)\ln {\epsilon}+81N\tilde{S}_{N}\left(N-\tilde{S}_{N}\right)\right\}},\nonumber \\
&&b=\frac{1}{3}\sqrt{9a^2-{{\left(4a+3\sqrt{N}\right)^2\ln{\epsilon}}\over{2N}}}.
\end{eqnarray}
On the other hand, from Eq.~(\ref{statement_2}) one can show that $\Lambda_{N}\overset{\epsilon}{<}\bar{K}^{\rm U}_{N,\epsilon}\left(S_{N}\right)$ for $\bar{K}^{\rm U}_{N,\epsilon}\left(S_{N}\right)=\left[\sqrt{N}S_{N}-N(a-b)\right]/(\sqrt{N}-2a)$, where $(a,b)$ is the solution to
\begin{eqnarray}\label{Kato_bar_U}
&&\argmin_{a,b}\hspace{.2cm}\frac{\sqrt{N}\tilde{S}_{N}-N(a-b)}{\sqrt{N}-2a}\nonumber\\
&&\textup{s.t.}\hspace{.3cm}\exp\left[\frac{-2(b^{2}-a^{2})}{\left(1-\displaystyle{\frac{4{}a}{3\sqrt{N}}}\right)^{2}}\right]=\epsilon,\hspace{.2cm}b\geq{}\abs{a},
\end{eqnarray}
given by~\cite{Navs}
\begin{eqnarray}\label{ab_U}
&&a={{3\sqrt{N}\left\{9\left(N-2\tilde{S}_{N}\right)\sqrt{N\ln{\epsilon}\left[N\ln{\epsilon}+18\tilde{S}_{N}\left(\tilde{S}_{N}-N\right)\right]}+4N\ln^{2}{\epsilon}+9\left(8\tilde{S}_{N}^2-8N\tilde{S}_{N}+3N^2\right)\ln{\epsilon}\right\}}\over{4\left\{4N\ln^{2}{\epsilon}+36\left(2\tilde{S}_{N}^2-2N\tilde{S}_{N}+N^2\right)\ln{\epsilon}+81N\tilde{S}_{N}\left(N-\tilde{S}_{N}\right)\right\}}},\nonumber \\
&&b={{\sqrt{18Na^2-\left(16a^2-24\sqrt{N}a+9N\right)\ln{\epsilon}}}\over{3\sqrt{2N}}}.
\end{eqnarray}
To finish with, let us point out a minor technical issue. Although we have established that $\Lambda_{N}\overset{\epsilon}{>}\bar{K}^{\rm L}_{N,\epsilon}\left(S_{N}\right)$ ($\Lambda_{N}\overset{\epsilon}{<}\bar{K}^{\rm U}_{N,\epsilon}\left(S_{N}\right)$), if only a lower (upper) bound on $S_{N}$ were known, say $S_{N}^{\rm L}$ ($S_{N}^{\rm U}$), it would be of interest to guarantee that $\Lambda_{N}\overset{\epsilon}{>}\bar{K}^{\rm L}_{N,\epsilon}\left(S_{N}^{\rm L}\right)$ ($\Lambda_{N}\overset{\epsilon}{<}\bar{K}^{\rm U}_{N,\epsilon}\left(S_{N}^{\rm U}\right)$). In fact, this is indeed the case as long as $\bar{K}^{\rm L}_{N,\epsilon}(x)$ ($\bar{K}^{\rm U}_{N,\epsilon}(x)$) is a non-decreasing function, which corresponds to the regime $a\geq{}-\sqrt{N}/2$ ($a\leq{}\sqrt{N}/2$). Namely, the desired inequalities hold if we replace $a$ in Eq.~(\ref{ab_L}) (Eq.~(\ref{ab_U})) by $a'=\max\left\{a,-\sqrt{N}/2\right\}$ $\left(a'=\min\left\{a,\sqrt{N}/2\right\}\right)$.
\section{Use case for Kato's inequality}\label{example}
To exemplify the application of Kato's inequality~\cite{Kato}, below we show that
\begin{equation}\label{desired_claim} N{}q_{\mathds{K}}\left\langle{1}\right\rangle_{\Omega_{j}^{\rm key}}Q_{j}^{\rm key}\overset{\epsilon}{>}K^{\rm L}_{N,\epsilon}(M_{j}^{\rm key})
\end{equation}
within the adversary model presented in Sec.~\ref{Assumptions} (note that Eq.~(\ref{desired_claim}) is just one of the multiple usages of Kato's inequality in the main text).

For this purpose, let
\begin{equation}\label{sequence}
\xi_{u}=\left\{
\begin{array}{ll}
1 & \rm{if}\hspace{.2cm}\text{Alice postselects $\sigma_{j}^{\rm key}$, Bob picks the key basis and a ``click" occurs in round $u$,}\\
0 & \rm{otherwise}, \\
\end{array} 
\right.
\end{equation}
for $u=1,2\ldots{}N$ (such that $\sum_{u=1}^{N}\xi_{u}=M_{j}^{\rm key}$), and let $\mathcal{F}_{0}\subseteq{}\mathcal{F}_{1}\subseteq{}\ldots{}\subseteq{}\mathcal{F}_{N}$ be the chain of $\sigma$-algebras induced by the sequence $R_{0}$, $(\xi_{1},R_{1})$, $(\xi_{2},R_{2}),\ \ldots{},\ (\xi_{N},R_{N})$. Here, we recall that $R_{u}$ denotes the state of Eve's classical register at the end of round $u$, and $R_{0}$ denotes an arbitrary initial state of the register (see Section~\ref{Assumptions}). From the definitions of $\xi_{u}$ and $\mathcal{F}_{u}$, the condition $E\left(\xi_{v}|\mathcal{F}_{u}\right)=\xi_{v}$ follows for all $v\leq{}u$, and therefore Kato's inequality applies. Hence, from Eq.~(\ref{claim}) we have that $\sum_{u=1}^{N}\Pr\left(\xi_{u}=1|\mathcal{F}_{u-1}\right)\overset{\epsilon}{>}K^{\rm L}_{N,\epsilon}\bigl(M_{j}^{\rm key}\bigr)$, such that the desired claim ---Eq.~(\ref{desired_claim})--- holds if
\begin{equation}\label{crucial}
\sum_{u=1}^{N}\Pr\left(\xi_{u}=1|\mathcal{F}_{u-1}\right)=N{}q_{\mathds{K}}\left\langle{1}\right\rangle_{\Omega_{j}^{\rm key}}Q_{j}^{\rm key}.
\end{equation}
This is what we establish next. On the one hand, by definition of $Q_{j}^{\rm key}$, the r.h.s. of Eq.~(\ref{crucial}) can be written as $\sum_{u=1}^{N}q_{\mathds{K}}\left\langle{1}\right\rangle_{\Omega_{j}^{\rm key}}p^{(u)}\bigl(\mathrm{click}|R_{u-1},\sigma_{j}^{\rm key},\mathrm{key}\bigr)$. As for the l.h.s., we have $\Pr\left(\xi_{u}=1|\mathcal{F}_{u-1}\right)=p^{(u)}\bigl(\sigma_{j}^{\rm key},\mathrm{key},\mathrm{click}|\mathcal{F}_{u-1}\bigr)=q_{\mathds{K}}\left\langle{1}\right\rangle_{\Omega_{j}^{\rm key}}p^{(u)}\bigl(\mathrm{click}|\mathcal{F}_{u-1},\sigma_{j}^{\rm key},\mathrm{key}\bigr)$, simply using the definition of $\xi_{u}$ and the fact that $p^{(u)}\bigl(\sigma_{j}^{\rm key},\mathrm{key}|\mathcal{F}_{u-1}\bigr)=q_{\mathds{K}}\left\langle{1}\right\rangle_{\Omega_{j}^{\rm key}}$. In short, Eq.~(\ref{crucial}) follows if
\begin{equation}\label{hidden_condition}
p^{(u)}\left(\mathrm{click}|\mathcal{F}_{u-1},\sigma_{j}^{\rm key},\mathrm{key}\right)=p^{(u)}\left(\mathrm{click}|R_{u-1},\sigma_{j}^{\rm key},\mathrm{key}\right).
\end{equation}
Remarkably though, this condition is true by hypothesis in our adversary model, since by assumption the only influence $\mathcal{F}_{u-1}$ may have on round $u$ is captivated by $R_{u-1}$.
%

\section{Channel model}\label{channel_model}
For the simulations of Sec.~\ref{performance}, the protocol observables are set to their expected values according to a standard channel model thoroughly described in~\cite{Zapatero}. For simplicity, the model disregards any possible misalignment occurring in the QKD link, and an active QKD receiver is assumed for comparison with prior work.

Given the overall efficiency of the system, $\eta$ (accounting for both channel and detection losses), and the dark count probability of Bob's detectors, $p_{\rm d}$, for all possible intensity settings $j$ the model yields
\begin{equation}
\mathds{E}\left[M_{j}^{\rm key}\right]=N{}q_{\mathds{K}}
\left\langle{1-(1-p_{\rm d})^{2}e^{-I\eta}}\right\rangle_{\Omega_{j}^{\rm key}},\hspace{.2cm}\mathds{E}\left[M_{j}^{\rm test}\right]=N{}q_{\mathds{T}}
\left\langle{1-(1-p_{\rm d})^{2}e^{-I\eta}}\right\rangle_{\Omega_{j}^{\rm test}}
\end{equation}
for the numbers of counts, and
\begin{equation}
\mathds{E}\left[m_{j}^{\rm test}\right]=2\times{}N{}q_{\mathds{T}}\ \Biggl\langle{\frac{1}{2}\left[1-(1-p_{\rm d})^{2}e^{\displaystyle{-I\eta}}\right]-\frac{1}{2}(1-p_{\rm d})\left[e^{\displaystyle{-\frac{I\eta\left(1-\sin\theta\cos\phi\right)}{2}}}-e^{\displaystyle{-\frac{I\eta\left(1+\sin\theta\cos\phi\right)}{2}}}\right]}\Biggr\rangle_{\Omega_{j}^{\rm H}}
\end{equation}
for the numbers of test-basis error counts. In a similar fashion, the size of the sifted keys, $M_{\rm key}=|\mathcal{X}^{\rm key}|$, and the corresponding number of bit errors, say $m_{\rm key}$, obey
\begin{equation}\label{sifted}
\mathds{E}\left[M_{\rm key}\right]=N{}q_{\mathds{K}}
\left\langle{1-(1-p_{\rm d})^{2}e^{-I\eta}}\right\rangle_{\Omega^{\rm key}}
\end{equation}
and
\begin{equation}\label{bit_errors}
\mathds{E}\left[m_{\rm key}\right]=2\times{}N{}q_{\mathds{K}}\Biggl\langle{\frac{1}{2}\left[1-(1-p_{\rm d})^{2}e^{\displaystyle{-I\eta}}\right]-\frac{1}{2}(1-p_{\rm d})\left[e^{\displaystyle{-I\eta\sin^{2}{\frac{\theta}{2}}}}-e^{\displaystyle{-I\eta\cos^{2}{\frac{\theta}{2}}}}\right]}\Biggr\rangle_{\Omega^{\rm R}},
\end{equation}
where we recall that $\Omega^{\rm key}=\bigcup_{j\in\Gamma^{\rm key}}\Omega_{j}^{\rm key}$ and we have also introduced $\Omega^{\rm R}=\bigcup_{j\in\Gamma^{\rm key}}\Omega^{\rm R}_{j}$ (note that Eq.~(\ref{bit_errors}) explicitly uses the fact that both polar caps contribute equally to the bit errors within the channel model). Importantly, $\mathds{E}\left[M_{\rm key}\right]$ and $\mathds{E}\left[m_{\rm key}\right]$ determine the model we use for the error correction leakage in the simulations via
\begin{equation}\label{leakage}
\lambda_{\rm EC}=f_{\rm EC}\mathds{E}\left[M_{\rm key}\right]h\left(\frac{\mathds{E}\left[m_{\rm key}\right]}{\mathds{E}\left[M_{\rm key}\right]}\right),
\end{equation}
where $f_{\rm EC}$ describes the efficiency of the error correction.

Finally, in this simple model, the targets of the linear programs (Eq.~(\ref{lp_1}) and Eq.~(\ref{lp_2})) fulfil
\begin{equation}\label{ideal_limit_1}
y_{\alpha,1}=1-(1-p_{\rm d})^{2}(1-\eta)
\end{equation}
and
\begin{equation}\label{ideal_limit_2}
e_{1}^{\rm ideal}=p_{\rm d}^{2}\frac{1}{2}+p_{\rm d}(1-p_{\rm d})(1-\eta)+p_{\rm d}(1-p_{\rm d})\eta{}\frac{1}{2}.
\end{equation}
While Eq.~(\ref{ideal_limit_1}) follows trivially, Eq.~(\ref{ideal_limit_2}) reflects the fact that, in the model, a perfect single-photon state (\textit{i.e.} $\ket{\rm H}$ or $\ket{\rm V}$) can only trigger an error if either both detectors experience a dark count ---which happens with a probability of $p_{\rm d}^{2}$--- or only the wrong detector experiences a dark count ---which happens with a probability of $p_{\rm d}(1-p_{\rm d})$---. Conditioned on the first event, the error probability is $1/2$, because double clicks are randomly assigned to a specific outcome. Conditioned on the second event, an error occurs if either the photon is lost (corresponding to the $(1-\eta)$ term) or the photon is not lost and thus fires the correct detector, but the resulting double click incurs in a bit error through the random assignment (corresponding to the $\eta/2$ term).

We remark that Eq.~(\ref{ideal_limit_1}) and Eq.~(\ref{ideal_limit_2}) provide the estimates that one would obtain in the perfect PE limit, which is plotted in Fig.~\ref{fig:performance} of the main text for comparison purposes.
\end{widetext}


%
%

\begin{thebibliography}{50}
\bibitem{comment}
Namely, for, say, the first interval, we are asserting that $\Pr\biggl[N{}q_{\mathds{K}}\left\langle{1}\right\rangle_{\Omega_{j}^{\rm key}}Q_{j}^{\rm key}\in\biggl(K^{\rm L}_{N,\epsilon}\bigl(M_{j}^{\rm key}\bigr),K^{\rm U}_{N,\epsilon}\bigl(M_{j}^{\rm key}\bigr)\biggr)\biggr]\geq{}1-2\epsilon$.
\bibitem{detail}
Note that Eq.~(\ref{Kato_LP}) trivially holds if we consider closed intervals instead, which we do for the LPs.
\end{thebibliography}
\end{document}